\documentclass[a4paper,12pt]{scrartcl}

\RequirePackage[utf8x]{inputenc}

\RequirePackage[english]{babel}
\RequirePackage[T1]{fontenc}

\RequirePackage{amsthm}
\RequirePackage{amsmath}
\RequirePackage{amssymb}
\RequirePackage{bbm}

\title{A Fundamental Solution to the Schrödinger Equation with Doss Potentials and its Smoothness}
\author{M. Grothaus and F. Riemann}
\date{March 17, 2015}

\theoremstyle{plain}
\newtheorem{theorem}{Theorem}[section]
\newtheorem{proposition}[theorem]{Proposition}
\newtheorem{lemma}[theorem]{Lemma}
\newtheorem{corollary}[theorem]{Corollary}

\theoremstyle{definition}

\newtheorem{example}[theorem]{Example}
\newtheorem{assumption}[theorem]{Assumption}
\newtheorem{remark}[theorem]{Remark}

\newcommand{\eps}{\varepsilon}

\newcommand{\C}{\mathbb{C}}
\newcommand{\R}{\mathbb{R}}
\newcommand{\N}{\mathbb{N}}

\newcommand{\cG}{\mathcal{G}}
\newcommand{\cP}{\mathcal{P}}

\newcommand{\rme}{\mathrm{e}}
\newcommand{\rmi}{\mathrm{i}}
\newcommand{\rmd}{\mathrm{d}}

\newcommand{\la}{\langle}
\newcommand{\ra}{\rangle}
\newcommand{\lla}{\la\!\la}
\newcommand{\rra}{\ra\!\ra}
\newcommand{\lv}{\lvert}
\newcommand{\rv}{\rvert}
\newcommand{\lV}{\lVert}
\newcommand{\rV}{\rVert}

\newcommand{\wh}{\widehat}
\newcommand{\ot}{\otimes}
\newcommand{\otk}{{\ot k}}
\newcommand{\otn}{{\ot n}}
\newcommand{\otm}{{\ot m}}
\newcommand{\wot}{\mathbin{\wh{\ot}}}
\newcommand{\wotn}{{\wot n}}

\newcommand{\wick}[2][n]{\mathopen{:}#2^{\ot #1}\mathclose{:}}
\newcommand{\wicko}[1][n]{\wick[#1]{\omega}}
\newcommand{\wickc}[1][n]{\mathopen{:}\,\cdot\,^{\ot #1}\mathclose{:}}

\newcommand{\indi}{\mathbbm{1}}

\DeclareMathOperator{\Id}{Id}
\DeclareMathOperator{\E}{\mathbb{E}}
\DeclareMathOperator{\spann}{span}
\DeclareMathOperator{\Real}{Re}
\DeclareMathOperator{\Imag}{Im}

\begin{document}

\maketitle

\begin{abstract}
\noindent
We construct a fundamental solution to the Schrödinger equation for a class of potentials of polynomial type by a complex scaling approach as in~\cite{Doss1980}. The solution is given as the generalized expectation of a white noise distribution. Moreover, we obtain an explicit formula as the expectation of a function of Brownian motion. This allows to show its differentiability in the classical sense. The admissible potentials may grow super-quadratically, thus by a result from~\cite{Yajima1996} the solution does not belong to the self-adjoint extension of the Hamiltonian.
\end{abstract}

\begin{section}{Introduction}

We consider the one-dimensional Schrödinger equation
\begin{align}\label{eq:schroedinger-equation}
\rmi\partial_tu(t,x) & = -\frac{1}{2}\Delta u(t,x)+V(x)u(t,x),\quad t>0,\,x\in\R,\\
u(0,x) & = u_0(x),\notag
\end{align}
with normalized mass and reduced Planck's constant, initial condition $u_0$ and a potential $V$ of a certain polynomial type specified below. Here $\Delta$ denotes the Laplace operator with respect to $x$.

For the free Schrödinger equation it is well-known that a fundamental solution is given by the free particle propagator
\begin{equation}\label{eq:free-propagator}
K_0(t,x;y):=\frac{1}{\sqrt{2\pi\rmi t}}\exp\Bigl(-\frac{1}{2\rmi t}(x-y)^2\Bigr)
\end{equation}
for $t>0$ and $x,y\in\C$. When $V\not=0$, various approaches by means of path integrals, see e.g.~\cite{DMMN79,Kleinert2009} and references therein, have been used in the past to give rigorous definitions of solutions for~\eqref{eq:schroedinger-equation}. For example via oscillatory integrals~\cite{AHKM08} or path integration in phase space~\cite{DMMN77,KD82}. The Feynman path integrand was constructed in~\cite{HS83a,HS83b} as a white noise distribution, and this technique was developed further in several ways, see e.g.~\cite{SS04,BG13} and the references therein. In~\cite{Doss1980} a complex scaling of the Feynman--Kac formula was employed.

The complex scaling method naturally involves analyticity conditions, and it was shown in~\cite{Doss1980} that for a certain class of analytic initial conditions $u_0\colon\C\to\C$ and potentials $V\colon\C\to\C$ a solution to~\eqref{eq:schroedinger-equation} is given by
\begin{equation}\label{eq:doss-solution}
(0,\infty)\times\R\ni(t,x)\longmapsto\E\biggl[u_0\bigl(x+\sqrt{\rmi}B_t\bigr)\exp\biggl(-\rmi\int_0^tV\bigl(x+\sqrt{\rmi}B_r\bigr)\,\rmd r\biggr)\biggr]\in\C,
\end{equation}
where $(B_t)_{t\ge 0}$ is a standard Brownian motion.

We are interested in finding a solution to~\eqref{eq:schroedinger-equation} for a more general class of initial conditions which also covers non-analytic states. To this end, we first seek a fundamental solution to~\eqref{eq:schroedinger-equation}, i.e.\ a function $(t,x,y)\mapsto K_V(t,x;y)$ which satisfies~\eqref{eq:schroedinger-equation} for fixed $y\in\R$ and whose initial condition is the Dirac delta distribution $\delta_x$, i.e.\ $\lim_{t\downarrow 0}K_V(t,x;y)=\delta_x(y)$ in the distributional sense. Replacing $u_0$ with $\delta_y$ in~\eqref{eq:doss-solution} we obtain the informal expression
\begin{equation}\label{eq:def-KV}
K_V(t,x;y)=\E\biggl[\delta_y\bigl(x+\sqrt{\rmi}B_t\bigr)\exp\biggl(-\rmi\int_0^tV\bigl(x+\sqrt{\rmi}B_r\bigr)\,\rmd r\biggr)\biggr].
\end{equation}
Such expressions have also been studied in~\cite{Westerkamp1995,Vogel2010,GSV12,GRS14}. We give a mathematically rigorous construction of~\eqref{eq:def-KV} and even of its integrand.

For $t>0$ and $x,y\in\R$ the object $\delta_y(x+\sqrt{\rmi}B_t)$ exists as a well-known white noise distribution and is called Donsker's delta, see e.g.~\cite{PT95}. However, existence of the product of Donsker's delta with the exponential term in~\eqref{eq:def-KV} has to be justified. A priori, products with Donsker's delta are only defined for a special class of testfunctions. Substantial progress for giving sense to this product was made in~\cite{Vogel2010}. There a representation using the Wick product and the translation and projection operator was provided. It can be seen that extensions of the multiplication with Donsker's delta are closely related to extensions of these operators. In Section~\ref{sec:whitenoiseanalysis} we realize unique extensions which allow to multiply Donsker's delta with a class of functions which depend on finitely many time points of Brownian motion. This approach was successfully used in~\cite{GRS14} to construct the path integral for the Edwards model of an electron in a random system of dense and weakly coupled scatterers.

Assuming that $V$ is of a certain polynomial type, it is possible to define
\begin{multline}\label{eq:riemann-approximation-limit}
\delta_y\bigl(x+\sqrt{\rmi}B_t\bigr)\exp\biggl(-\rmi\int_0^tV\bigl(x+\sqrt{\rmi}B_r\bigr)\,\rmd r\biggr)\\
:=\lim_{n\to\infty}\delta_y\bigl(x+\sqrt{\rmi}B_t\bigr)\exp\biggl(-\rmi\sum_{[r,s]\in Q_n}(s-r)V\bigl(x+\sqrt{\rmi}B_r\bigr)\biggr)
\end{multline}
as a white noise distribution, where $(Q_n)_{n\in\N}$ is a sequence of partitions of the interval $[0,t]$ whose mesh converges to zero, i.e.\ we approximate the Riemann-Integral by Riemann sums and~\eqref{eq:riemann-approximation-limit} does not depend on the particular choice of $(Q_n)_{n\in\N}$. Applying the generalized expectation of white noise distributions to this object defines~\eqref{eq:def-KV}. In the end we obtain the explicit probabilistic representation
\begin{equation}\label{eq:KV-representation-intro}
K_V(t,x;y)=K_0(t,x;y)\cdot\E\biggl[\exp\biggl(-\rmi\int_0^tV\Bigl(y+\frac{r}{t}(x-y)+\sqrt{\rmi}B_r-\frac{r}{t}\sqrt{\rmi}B_t\Bigr)\,\rmd r\biggr)\biggr].
\end{equation}
We show that $K_V$ solves the Schrödinger equation for fixed $y\in\R$ with initial condition $\lim_{t\downarrow 0}K_V(t,x;y)=\delta_x(y)$. To consider some initial state $u_0$ we define
\begin{equation}\label{eq:def-solution-via-integral}
u(t,x):=\int_\R K_V(t,x;y)u_0(y)\,\rmd y
\end{equation}
for $t>0$, $x\in\R$, and prove that $u$ solves~\eqref{eq:schroedinger-equation} with $\lim_{t\downarrow 0}u(t,x)=u_0(x)$ whenever $u_0\colon\R\to\C$ is compactly supported and twice differentiable. This enlarges the class of admissible initial conditions from~\cite{Doss1980}, since the only analytic and compactly supported function is constantly zero.

The class of potentials contains super-quadratic polynomials such as $V(x)=x^6$. A general result from~\cite{Yajima1996} states that for such super-quadratic potentials the fundamental solution to~\eqref{eq:schroedinger-equation}, obtained via the unique self-adjoint extension of the Hamiltonian, cannot be classically differentiable. However, we are able to show existence and continuity of the partial derivatives of $K_V$, thus the fundamental solution obtained here does not belong to this self-adjoint extension. This shows that the path integral approach may lead to solutions smoother than those obtained via the semigroup approach.

Let us summarize the results obtained in the present article.
\begin{enumerate}
\item We construct the fundamental solution $K_V$ for a class of polynomial potentials specified in Assumption~\ref{ass:potential} below, see Theorem~\ref{thm:KV-construction}.
\item We rigorously derive the explicit formula~\eqref{eq:KV-representation-intro} for $K_V$ in terms of the expectation of a function of Brownian motion. This formula might be useful for numerical simulations of the fundamental solution.
\item Formula~\eqref{eq:KV-representation-intro} allows to show smoothness of $K_V$, see Corollary~\ref{cor:KV-is-differentiable}. In combination with~\cite{Yajima1996} this proves that $K_V$ does not coincide with the solution constructed via unitary groups of operators.
\item It is verified that $K_V$ is a fundamental solution to the Schrödinger equation, see Theorem~\ref{thm:KV-is-fundamental-solution}.
\item The constructed fundamental solution enables us to treat initial states with compact support by the ansatz in~\eqref{eq:def-solution-via-integral}, see Theorem~\ref{thm:u-solves-schroedinger-equation}. This is of interest from the physical point of view and goes beyond the results of~\cite{Doss1980}.
\end{enumerate}

\end{section}

\begin{section}{Class of Potentials}

We fix some notations used in this article. For topological spaces $X$ and $Y$ we denote the set of continuous functions from $X$ to $Y$ by $C(X;Y)$. For any set $M$ and a bounded function $f\colon M\to\C$ we let $\lV f\rV_\infty:=\sup_{m\in M}\lv f(m)\rv$. For a Brownian motion $(B_t)_{t\ge 0}$ and $T\ge 0$ we define the random variable $\lV B\rV_T:=\sup_{t\in[0,T]}\lv B_t\rv$.

\begin{assumption}\label{ass:potential}
We assume that the potential $V$ is a polynomial of the following type:
\[ \C\ni z\longmapsto V(z)=(-\rmi)^{N+1}a_{2N}z^{2N}+\sum_{k=0}^{2N-1}a_kz^k\in\C \]
for some $N\in\N$, $a_{2N}>0$ and $a_0,\dotsc,a_{2N-1}\in\C$. Such polynomials have also been considered in~\cite{Doss1980,Vogel2010,GSV12}.
\end{assumption}

This class contains the complex polynomials
\begin{align*}
z & \mapsto -z^2, & z & \mapsto \rmi z^4, & z & \mapsto z^6, & z & \mapsto -\rmi z^8,\\
z & \mapsto -z^{10}, & z & \mapsto \rmi z^{12}, & z & \mapsto z^{14}, & z & \mapsto -\rmi z^{16},
\end{align*}
and so on. Note that the polynomials from this class which fulfill $V(\R)\subset\R$ must satisfy $N\in2\N+1$.

\begin{lemma}\label{lem:boundedness-of-exponential-term}
For all $R\in[0,\infty)$ there exists some constant $c\in[0,\infty)$ such that
\[ \sup_{\substack{f\in C([0,t],\C)\\\lV f\rV_\infty\le R}}\,\sup_{g\in C([0,t],\R)}\biggl\lv\exp\biggl(-\rmi\int_0^tV\bigl(f(r)+\sqrt{\rmi}g(r)\bigr)\,\rmd r\biggr)\biggr\rv\le\rme^{ct} \]
for all $t\ge 0$.
\end{lemma}

\begin{proof}
By binomial expansion it is obvious that
\[ V\bigl(z+\sqrt{\rmi}w\bigr)=-\rmi a_{2N}w^{2N}+\sum_{k=0}^{2N-1}p_k(z)w^k \]
for some polynomials $p_0,\dotsc,p_{2N-1}$ and all $z,w\in\C$. Define
\[ c':=\max_{k=0,\dotsc,2N-1}\max_{\lv z\rv\le R}\lv p_k(z)\rv\in[0,\infty). \]
By domination of the leading power we have
\[ c:=\sup_{w\in\R}\biggl(-a_{2N}w^{2N}+c'\sum_{k=0}^{2N-1}\lv w\rv^k\biggr)\in[0,\infty). \]
Hence for all $z\in\C$ with $\lv z\rv\le R$ and $w\in\R$ it holds
\[ \Imag V\bigl(z+\sqrt{\rmi}w\bigr)=-a_{2N}w^{2N}+\Imag\sum_{k=0}^{2N-1}p_k(z)w^k\le-a_{2N}w^{2N}+c'\sum_{k=0}^{2N-1}\lv w\rv^k\le c. \]
It follows
\[ \biggl\lv\exp\biggl(-\rmi\int_0^tV\bigl(f(r)+\sqrt{\rmi}g(r)\bigr)\,\rmd r\biggr)\biggr\rv=\exp\biggl(\int_0^t\Imag V\bigl(f(r)+\sqrt{\rmi}g(r)\bigr)\,\rmd r\biggr)\le\rme^{ct} \]
for all $t\ge 0$, $f\in C([0,t],\C)$ with $\lV f\rV_\infty\le R$ and $g\in C([0,t],\R)$.
\end{proof}

The following is easy to show, using binomial expansion.

\begin{lemma}\label{lem:uniform-polynomial-bound}
Let $A\subset\C$ be bounded and $P\colon\C\to\C$ be an arbitrary polynomial. Then there exists $C\in\R$ such that for all $b\in\C$ it holds
\[ \sup_{a\in A}\lv P(a+b)\rv\le C\sum_{k=0}^{\deg P}\lv b\rv^k. \]
\end{lemma}

For a proof of the following fact see for example~\cite[Lem.~1]{Doss1980}.

\begin{lemma}\label{lem:doss-expectation-estimate}
For a polynomial $P\colon\C\to\C$ and $0\le T<\infty$ it holds $P(\lV B\rV_T)\in L^1(\mu)$.
\end{lemma}

From the free case $V\equiv 0$ of~\cite[Thm.~3]{Doss1980} one obtains

\begin{theorem}\label{thm:freedoss}
Let $0<T<\infty$ and $u_0\colon\C\to\C$ be analytic such that for all $z\in\C$ the random variables
\[ \sup_{t\in[0,T]}\bigl\lv u_0\bigl(z+\sqrt{\rmi}B_t\bigr)\bigr\rv\quad\text{and}\quad\sup_{t\in[0,T]}\bigl\lv\Delta u_0\bigl(z+\sqrt{\rmi}B_t\bigr)\bigr\rv \]
are integrable. Furthermore, assume that $\C\ni z\mapsto\E\bigl[u_0\bigl(z+\sqrt{\rmi}B_t\bigr)\bigr]\in\C$ is analytic for all $t\in[0,T]$. Then \[ [0,T]\times\C\ni(t,z)\longmapsto u(t,z):=\E\bigl[u_0\bigl(z+\sqrt{\rmi}B_t\bigr)\bigr]\in\C \]
solves the Schrödinger equation~\eqref{eq:schroedinger-equation} in $(t,x)\in[0,T]\times\R$ for $V\equiv 0$ with initial condition $u(0,\cdot)=u_0$.
\end{theorem}

\end{section}

\begin{section}{White Noise Analysis}\label{sec:whitenoiseanalysis}

To give a mathematically rigorous meaning to the expression in~\eqref{eq:def-KV} we use the framework of white noise analysis, which has been applied to construct path integrals, in particular Feynman integrals, in several articles, see~\cite{HS83a,HS83b,Kuo83} and also~\cite{SS04} and the references therein. Rich background to the theory can be found in the monographs~\cite{HKPS93,Kuo96,Obata94}.

In this section we briefly recap the construction of the white noise measure $\mu$ and the famous decomposition of the space $L^2(\mu)$ into orthogonal polynomials, known as Wiener-It\={o} decomposition or chaos decomposition. Afterwards we present the definition of the spaces $\cG$ and $\cG'$ of regular testfunctions and distributions, which were introduced in~\cite{PT95}. We repeat the definitions of the translation and projection operators and provide extensions which allow to define~\eqref{eq:riemann-approximation-limit} as an element of $\cG'$. Then~\eqref{eq:def-KV} is defined as the generalized expectation of~\eqref{eq:riemann-approximation-limit}.

\begin{subsection}{The White Noise Measure}

Let $L^2(\R)$ denote the space of (equivalence classes of) real-valued square-integrable functions on the real line with respect to the Lebesgue measure, equipped with its usual inner product $(\cdot,\cdot)$ and corresponding norm $\lv\cdot\rv$. The Schwartz space of rapidly decreasing functions is denoted by $S(\R)$ and equipped with its usual nuclear topology, see e.g.~\cite{Kuo96} or~\cite[Ch.~V,~Sec.~3]{RS72}. Its dual space of tempered distributions we denote by $S'(\R)$. By identifying $L^2(\R)$ with its dual via the Riesz isomorphism, the canonical dual pairing $\la\cdot,\cdot\ra$ on $S'(\R)\times S(\R)$ is realized as an extension of the inner product in $L^2(\R)$, so $\la\xi,\zeta\ra=(\xi,\zeta)$ for $\xi\in L^2(\R)$ and $\zeta\in S(\R)$. In particular $S(\R)\subset L^2(\R)\subset S'(\R)$.

The standard Gaussian measure $\mu$ on $S'(\R)$, equipped with its cylindrical $\sigma$-algebra, arises from its characteristic function
\[ \int_{S'(\R)}\exp(\rmi\la\omega,\xi\ra)\,\rmd\mu(\omega)=\exp\Bigl(-\frac{1}{2}\la\xi,\xi\ra\Bigr),\quad\xi\in S(\R), \]
via the Bochner--Minlos theorem. For $1\le p\le\infty$ we abbreviate $L^p(\mu):=L^p(S'(\R),\C;\mu)$ together with its usual norm $\lV\cdot\rV_{L^p(\mu)}$. A fundamental property of the measure $\mu$ is that for fixed $\xi_1,\dotsc,\xi_d\in S(\R)$, $d\in\N$, the random vector $(\la\cdot,\xi_1\ra,\dotsc,\la\cdot,\xi_d\ra)$ has a centered Gaussian distribution on $\R^d$ with covariance matrix $(\la\xi_k,\xi_l\ra)_{k,l=1,\dotsc,d}$. Thus, if $\la\cdot,\cdot\ra$ is extended in a bilinear way to elements from the complexified spaces (which we tag with a subscript $\C$), one obtains that the space of smooth polynomials
\[ \cP:=\spann\{\la\cdot,\xi\ra^n:\xi\in S_\C(\R),n\in\N\} \]
is a subspace of $L^2(\mu)$.

We use the notation $L_\C^2(\R)^\wotn$ for the $n$-fold symmetric Hilbert space tensor product of $L_\C^2(\R)$ and keep the symbol $\la\cdot,\cdot\ra$ for its bilinear dual pairing and $\lv\cdot\rv$ for its norm. Similarly as before, the dual pairing between $S_\C(\R)^\wotn$ and its dual space is realized as a bilinear extension of $\la\cdot,\cdot\ra$ and denoted by the same symbol. Here $S(\R)^\wotn$ denotes the subspace of  $S(\R^n)$ consisting of symmetric Schwartz functions on $\R^n$. With this notation, each $\varphi\in\cP$ of degree $N\in\N$ can uniquely be represented as a Wick polynomial
\begin{equation}\label{eq:smooth-polynomial-chaos-decomposition}
\varphi=\sum_{n=0}^N\la\wickc,\varphi^{(n)}\ra,\quad\varphi^{(n)}\in\spann\{\xi^\otn:\xi\in S_\C(\R)\},\,n=0,\dotsc,N,
\end{equation}
where $\wicko$ denotes the $n$-th Wick power of $\omega\in S'(\R)$ and $\varphi^{(n)}$ is called the $n$-th kernel of $\varphi$, see e.g.~\cite{HKPS93,Kuo96,Obata94} for details. The advantage of this representation is the orthogonality relation
\[ \int_{S'(\R)}\bigl\la\wicko,\varphi^{(n)}\bigr\ra\bigl\la\wicko[m],\psi^{(m)}\bigr\ra\,\rmd\mu(\omega)=\delta_{nm}n!\bigl\la\varphi^{(n)},\psi^{(m)}\bigr\ra, \]
where $\delta_{nm}$ denotes the Kronecker delta. This implies $\lV\la\wickc,\varphi^{(n)}\ra\rV_{L^2(\mu)}^2=n!\lv\varphi^{(n)}\rv^2$ for $\varphi^{(n)}$ as in~\eqref{eq:smooth-polynomial-chaos-decomposition} and thus $\la\wickc,f^{(n)}\ra$ can be defined as a limit in $L^2(\mu)$ for general $f^{(n)}\in L_\C^2(\R)^\wotn$ by density of such $\varphi^{(n)}$ in $L_\C^2(\R)^\wotn$. In particular $\la\cdot,\xi\ra\in L^2(\mu)$ for $\xi\in L_\C^2(\R)$. One can show that the space of smooth polynomials is dense in $L^2(\mu)$, which implies that for every $F\in L^2(\mu)$ there exist unique $f^{(n)}\in L_\C^2(\R)^\wotn$, $n\in\N$, such that $F=\sum_{n=0}^\infty\la\wickc,f^{(n)}\ra$ holds in $L^2(\mu)$. This expansion is called Wiener-It\={o}-decomposition and $f^{(n)}$ is called the $n$-th kernel of $F$. Its norm in $L^2(\mu)$ can be computed as $\lV F\rV_{L^2(\mu)}^2=\sum_{n=0}^\infty n!\lv f^{(n)}\rv^2$.

It is still valid that the random vector $(\la\cdot,\xi_1\ra,\dotsc,\la\cdot,\xi_d\ra)$ has a centered Gaussian distribution with covariance matrix $(\la\xi_k,\xi_l\ra)_{k,l=1,\dotsc,d}$ when $\xi_1,\dotsc,\xi_d\in L^2(\R)$. In particular the process $(\la\cdot,\indi_{[0,t)}\ra)_{t\ge 0}$ in $L^2(\mu)$ has the law of Brownian motion and the Kolmogorov-\u{C}ensov-Lo\`eve theorem ensures that this process has a modification with continuous paths, which we denote by $(B_t)_{t\ge 0}$.

\end{subsection}

\begin{subsection}{Regular Testfunctions and Distributions}

The space of regular testfunctions $\cG$ and the space of regular distributions $\cG'$ were introduced in~\cite{PT95}. We use these spaces to define the object in~\eqref{eq:riemann-approximation-limit}. For convenience we repeat the definition and some important properties of these spaces. For $q\in\R$ let the norm $\lV\cdot\rV_q$ on $\cP$ be defined by
\[ \lV\varphi\rV_q^2:=\sum_{n=0}^N n!2^{q n}\bigl\lv\varphi^{(n)}\bigr\rv^2, \]
where $\varphi\in\cP$ is as in~\eqref{eq:smooth-polynomial-chaos-decomposition}. The space $\cG_q$ is defined to be the completion of $\cP$ with respect to $\lV\cdot\rV_q$. One easily sees $\cG_q\subset\cG_r$ continuously for $q\ge r$ and $\cG_0=L^2(\mu)$. Every $\varphi\in\cG_q$ admits a unique representation
\begin{equation}\label{eq:regular-testfunction-chaos-decomposition}
\varphi=\sum_{n=0}^\infty\la\wickc,\varphi^{(n)}\ra,\quad\varphi^{(n)}\in L_\C^2(\R)^\wotn, n\in\N,
\end{equation}
where the series converges in $\cG_q$. Then $\varphi^{(n)}$ is called the $n$-th kernel of $\varphi$. Note that in case $q$ is negative, this expression may not be considered a pointwisely defined object despite the notation. The space $\cG$ is defined by
\[ \cG:=\bigcap_{q\in\R}\cG_q \]
and equipped with the projective limit topology, i.e.\ the coarsest topology such that all inclusions $\cG\subset\cG_q$ are continuous. There exists a complete metric which induces this topology, and a sequence converges in $\cG$ if and only if it converges in each of the spaces $\cG_q$. In particular $\cG\subset L^2(\mu)$ densely and continuously and each $\varphi\in\cG$ can be represented uniquely as in~\eqref{eq:regular-testfunction-chaos-decomposition} where the series converges in $\cG$, that is in each $\cG_q$.

For $q\ge 0$ the space $\cG_{-q}$ is actually the dual space of $\cG_q$ with respect to the dual pairing $\lla\cdot,\cdot\rra$ on $\cG_{-q}\times\cG_q$ given by $\lla\Phi,\varphi\rra:=\sum_{n=0}^\infty n!\la\Phi^{(n)},\varphi^{(n)}\ra$ for $\Phi\in\cG_{-q}$ and $\varphi\in\cG_q$. This dual pairing extends the inner product on the real part of $L^2(\mu)$ bilinearly in the sense that $\lla\Phi,\varphi\rra=\int_{S'(\R)}\Phi\varphi\,\rmd\mu$ if $\Phi\in L^2(\mu)$. It follows from general duality theory that
\[ \cG'=\bigcup_{q\in\R}\cG_q. \]
In~\cite{PT95} no topology was considered on $\cG'$, however this representation allows to choose the inductive limit topology on $\cG'$, that is the finest locally convex topology such that all inclusions $\cG_q\subset\cG'$, $q\in\R$, are continuous. Since $\cG_0=L^2(\mu)$ it is obvious that $\cG\subset L^2(\mu)\subset\cG'$ continuously. More generally it was shown in~\cite{PT95} that $\cG\subset L^p(\mu)$ continuously for all $1<p<\infty$, which by duality implies that even $\cG\subset L^p(\mu)\subset\cG'$ continuously. As before $\lla\Phi,\varphi\rra=\int\Phi\varphi\,\rmd\mu$ for $\Phi\in L^p(\mu)$, $\varphi\in\cG$.

The \emph{$S$-transform} of $\Phi\in\cG'$ is defined as the mapping
\[ L_\C^2(\R)\ni\xi\longmapsto S\Phi(\xi):=\sum_{n=0}^\infty\la\Phi^{(n)},\xi^\otn\ra\in\C. \]
It is a holomorphic function and completely characterizes $\Phi$. The indicator function $\indi_{S'(\R)}$ on $S'(\R)$ is a smooth polynomial, thus the generalized expectation of $\Phi\in\cG'$ can be defined as $\E[\Phi]:=\lla\Phi,\indi_{S'(\R)}\rra=S\Phi(0)\in\C$, which extends the usual expectation on $L^2(\mu)$. An important operation is the so-called Wick product, which is a continuous bilinear mapping from $\cG'\times\cG'$ to $\cG'$ characterized by $S(\Phi\diamond\Psi)(\xi)=S\Phi(\xi)\cdot S\Psi(\xi)$ for $\Phi,\Psi\in\cG'$ and $\xi\in L_\C^2(\R)$. In particular $\E[\Phi\diamond\Psi]=\E[\Phi]\cdot\E[\Psi]$.

Another property of $\cG$ is that it is closed under pointwise multiplication which is a continuous bilinear operator from $\cG\times\cG$ to $\cG$. One can extend this multiplication, allowing one factor to be in $\cG'$ by defining $\lla\Phi\cdot\varphi,\psi\rra:=\lla\Phi,\varphi\cdot\psi\rra$ for $\Phi\in\cG'$ and $\varphi,\psi\in\cG$ and this multiplication is separately continuous from $\cG'\times\cG$ to $\cG'$.

\begin{example}
A well-established regular distribution is \emph{Donsker's delta} $\delta(z\la\cdot,\eta\ra-a)$, defined for $a,z\in\C$ with $\Real z>0$ and $\eta\in L^2(\R)\setminus\{0\}$ and characterized via its $S$-transform
\[ S(\delta(z\la\cdot,\eta\ra-a))(\xi)=\frac{1}{\sqrt{2\pi z^2\la\eta,\eta\ra}}\exp\biggl(-\frac{1}{2z^2\la\eta,\eta\ra}\bigl(a-z\la\xi,\eta\ra\bigr)^2\biggr),\quad\xi\in L_\C^2(\R). \]
First definitions only covered the real case, i.e.\ $a,z\in\R$, see e.g.~\cite{HKPS93,Kuo96,PT95}. Generalizing via complex scaling yields the definition for complex parameters, see~\cite{LLSW94}. It can be considered as the formal composition of the Dirac delta distribution with $z\la\cdot,\eta\ra$. From uniqueness of the $S$-transform one easily sees that Donsker's delta is homogeneous of degree $-1$, in the sense that
\[ \delta(\la\cdot,\eta\ra-a)=z\cdot\delta(z\la\cdot,\eta\ra-za) \]
for $\eta\in L^2(\R)\setminus\{0\}$ and $a,z\in\C$ with $\Real z>0$. To realize~\eqref{eq:def-KV} we choose $\eta=\indi_{[0,t)}$ and $z=\sqrt{\rmi}=(1+\rmi)/\sqrt{2}$. Note that $\E[\delta(\sqrt{\rmi}\la\cdot,\indi_{[0,t)}\ra-(y-x))]=K_0(t,x;y)$.
\end{example}

\end{subsection}

\begin{subsection}{Translation, Projection, and Donsker's Delta}

In this section we show that there is a unique continuous extension of the pointwise product with Donsker's delta to subspaces of $L^2(\mu)$ whose members only depend on finitely many monomials. This extension involves the operators of translation and projection, which themselves have to be extended in an appropriate manner. Furthermore, we show a convenient method to find representations of these extensions, which is necessary for applications. This extends results from~\cite{GRS14}, where the case $a=0$, i.e.\ multiplication with $\delta(z\la\cdot,\eta\ra)$ was considered.

For $\xi=(\xi_1,\dotsc,\xi_d)\in L^2(\R)^d$ let
\[ \cP_\xi:=\{P(\la\cdot,\xi_1\ra,\dotsc,\la\cdot,\xi_d\ra),P\colon\C^d\to\C\text{ is a polynomial}\} \]
be the space of polynomials which only depend on the finitely many monomials $\la\cdot,\xi_i\ra$, $i=1,\dotsc,d$. Since $\la\cdot,\xi_i\ra\in\cG$ for each $i=1,\dotsc,d$ and $\cG$ is closed with respect to pointwise multiplication, it follows that $\cP_\xi\subset\cG$. It can be shown using the well-known polarization formula that for the kernels $\varphi^{(n)}$ of $\varphi\in\cP_\xi$ it holds that each $\varphi^{(n)}$ is a linear combination of elements of the form $\zeta^\otn$, where $\zeta\in\spann_\C\{\xi_1,\dotsc,\xi_d\}$.

For $p\in[1,\infty)$ let $\overline{\cP_\xi}^{L^p(\mu)}$ be the closure of $\cP_\xi$ in $L^p(\mu)$. Since $(\la\cdot,\xi_1\ra,\dotsc,\la\cdot,\xi_d\ra)$ has a centered Gaussian distribution with covariance matrix $(\la\xi_k,\xi_l\ra)_{k,l=1,\dotsc,d}$, we obtain the following characterization.

\begin{lemma}\label{lem:finite-dimensional-Lp-isometry}
Let $\xi_1,\dotsc,\xi_d\in L^2(\R)$ be linearly independent and consider the centered Gaussian measure $\mu_M=\mathcal{N}(0,M)$ on $\R^d$ with covariance matrix $M=(\la\xi_k,\xi_l\ra)_{k,l=1,\dotsc,d}$. Then
\[ L^p(\R^d,\mu_M;\C)\ni f\longmapsto f(\la\cdot,\xi_1\ra,\dotsc,\la\cdot,\xi_d\ra)\in L^p(\mu) \]
is an isometry. Its range is given by
\[ \overline{\cP_\xi}^{L^p(\mu)}=\bigl\{f(\la\cdot,\xi_1\ra,\dotsc,\la\cdot,\xi_d\ra):f\in L^p(\R^d,\mu_M;\C)\bigr\} \]
due to density of the polynomials in $L^p(\R^d,\mu_M;\C)$.
\end{lemma}

A notation used below is the \emph{symmetric contraction of tensor products}. For $k,n,m\in\N$ with $k\le n\wedge m$ it is the continuous bilinear operator $\wot_k\colon L_\C^2(\R)^\wotn\times L_\C^2(\R)^{\wot m}\to L_\C^2(\R)^{\wot n+m-2k}$ determined by the property
\[ \xi^\otn\wot_k\zeta^\otm=\la\xi,\zeta\ra^k\xi^{\ot n-k}\wot\zeta^{\ot m-k} \]
for $\xi,\zeta\in L^2(\R)$. Here $\xi^{\ot n-k}\wot\zeta^{\ot m-k}$ denotes the symmetrization of $\xi^{\ot n-k}\ot\zeta^{\ot m-k}$. The contraction of tensor products is in fact a contraction, thus
\[ \bigl\lv\varphi^{(n)}\wot_k\psi^{(m)}\bigr\rv\le\bigl\lv\varphi^{(n)}\bigr\rv\cdot\bigl\lv\psi^{(m)}\bigr\rv \]
for $\varphi^{(n)}\in L_\C^2(\R)^\wotn$, $\psi^{(m)}\in L_\C^2(\R)^{\wot m}$.

\begin{subsubsection}{Translation Operator}

The translation operator $T_\eta$, $\eta\in S'(\R)$, is well-known in Gaussian analysis, see for example~\cite{Kuo96,PT95}. Intuitively it assigns to a function $\varphi\colon S'(\R)\to\C$ the shifted function $\varphi(\cdot+\eta)$. It has been shown in~\cite{PT95} that $T_\eta$ acts continuously on $\cG$ for $\eta\in L^2(\R)$. We are interested in a definiton of this operator for \emph{complex} $\eta\in L_\C^2(\R)$. The problem one encounters in a pointwise definition in this situation is that the expression $\varphi(\cdot+\eta)$ may not be well-defined for $\varphi\in\cG$, since by definition such $\varphi$ is an equivalence class of functions defined only on the real space $S'(\R)$. In~\cite{Westerkamp1995} a definition of $T_\eta$ for complex $\eta$ was given in terms of the chaos decomposition and it was shown that this operator acts continuously on $\cG$. We choose the following approach: Every $\varphi\in\cP$ is an analytic function on $S_\C'(\R)$, so the definition $T_\eta\varphi:=\varphi(\cdot+\eta)$ is possible in that case. Furthermore, it is known that $\varphi(\cdot+\eta)$ is again a smooth polynomial with chaos expansion
\begin{equation}\label{eq:chaos-decomposition-of-translation}
T_\eta\varphi=\varphi(\cdot+\eta)=\sum_{n=0}^N\sum_{k=0}^{N-n}\binom{n+k}{k}\bigl\la\wickc,\eta^\otk\wot_k\varphi^{(n+k)}\bigr\ra
\end{equation}
for $\varphi\in\cP$ as in~\eqref{eq:smooth-polynomial-chaos-decomposition}. Below we prove suitable norm-estimates which show that this operator extends continuously to a linear operator on $\cG'$, so in particular it is defined on $L^2(\mu)$. This extension coincides with the operator defined in~\cite{Westerkamp1995} on $\cG$. It is not clear under which conditions on $\varphi$ the identity $T_\eta\varphi=\varphi(\cdot+\eta)$ is valid for this extension, even when $\varphi\colon S_\C'(\R)\to\C$ is pointwisely defined. However, often one has an intuitive understanding of the object $T_\eta\varphi$. We present a method which can prove correctness of an intuitively obtained representation.

A statement similar to the one below was already proven in~\cite[Thm.~67]{Westerkamp1995}. We drop the assumption $q,r>0$ set there and improve the upper bound for the operator norm.

\begin{proposition}\label{prop:translationoperator-continuous}
Let $\eta\in L_\C^2(\R)$ and $q,r\in\R$ with $r<q$. Then there exists $C_1=C_1(q,r,\lv\eta\rv)\in\R$ such that for all $\varphi\in\cP$ it holds $\lV T_\eta\varphi\rV_r\le C_1\,\lV\varphi\rV_q$. In particular, $T_\eta$ extends uniquely to a bounded linear operator from $\cG_q$ to $\cG_r$.
\end{proposition}

\begin{proof}
Let $\varphi\in\cP$ as in~\eqref{eq:smooth-polynomial-chaos-decomposition} and denote the $n$-th kernel of $T_\eta\varphi$ by $\tilde{\varphi}^{(n)}$, $n=0,\dotsc,N$, which are given in~\eqref{eq:chaos-decomposition-of-translation}. We estimate
\begin{align*}
\bigl\lv\tilde{\varphi}^{(n)}\bigr\rv^2
 & \le \biggl(\sum_{k=0}^{N-n}\frac{(n+k)!}{n!k!}\lv\eta\rv^k\bigl\lv\varphi^{(n+k)}\bigr\rv\biggr)^2\\
 & \le \biggl(\sum_{k=0}^{N-n}(n+k)!2^{q(n+k)}\bigl\lv\varphi^{(n+k)}\bigr\rv^2\biggr)\cdot\biggl(\sum_{k=0}^{N-n}\frac{(n+k)!}{n!n!k!k!}2^{-q(n+k)}\lv\eta\rv^{2k}\biggr)\\
 & \le \lV\varphi\rV_q^2\cdot\frac{1}{n!}2^{-q n}\sum_{k=0}^{N-n}\binom{n+k}{k}\frac{1}{k!}2^{-q k}\lv\eta\rv^{2k}.
\end{align*}
Thus a rearrangement of summation yields
\begin{align*}
\lV T_\eta\varphi\rV_r^2 = \sum_{n=0}^Nn!2^{nr}\bigl\lv\tilde{\varphi}^{(n)}\bigr\rv^2
 & \le \lV\varphi\rV_q^2\sum_{n=0}^N2^{(r-q)n}\sum_{k=0}^{N-n}\binom{n+k}{k}\frac{1}{k!}2^{-q k}\lv\eta\rv^{2k}\\
 & = \lV\varphi\rV_q^2\sum_{k=0}^N\frac{1}{k!}2^{-q k}\lv\eta\rv^{2k}\sum_{n=0}^{N-k}\binom{n+k}{k}2^{(r-q)n}\\
 & \le \lV\varphi\rV_q^2\sum_{k=0}^N\frac{1}{k!}2^{-q k}\lv\eta\rv^{2k}\bigl(1-2^{r-q}\bigr)^{-(k+1)}\\
 & \le \lV\varphi\rV_q^2\cdot\bigl(1-2^{r-q}\bigr)^{-1}\exp\biggl(\frac{\lv\eta\rv^2}{2^q-2^r}\biggr),
\end{align*}
where we used the generating function
\[ \sum_{n=0}^\infty\binom{n+k}{k}x^n=(1-x)^{-(k+1)} \]
of the binomial sequence, convergent for all $k\in\N$ and $\lv x\rv<1$.
\end{proof}

Evidently, the extensions of $T_\eta$ provided by this theorem mutually coincide on their common domain, hence we denote each of them by $T_\eta$ without danger of confusion. Since $\cG'=\bigcup_{q\in\R}\cG_q$ the theorem implies that $T_\eta$ is a linear operator from $\cG'$ to $\cG'$. By~\cite[Thm.~6.1]{Schaefer71} it holds that $T_\eta\cG'\to\cG'$ is continuous if and only if every restriction $T_\eta\colon\cG_q\to\cG'$, $q\in\R$, is continuous. Then the following is immediate.

\begin{corollary}\label{cor:translationoperator-continuous}
The linear operator $T_\eta$ is continuous from $\cG'$ to $\cG'$ and from $\cG$ to $\cG$.
\end{corollary}

The translation operator has some expected properties.

\begin{lemma}
For $\eta,\xi\in L_\C^2(\R)$ and $\Phi\in\cG'$, $\varphi\in\cG$ the following hold: $T_\eta T_\xi\Phi=T_{\eta+\xi}\Phi$, $S(T_\eta\Phi)(\xi)=S\Phi(\eta+\xi)$ and $T_\eta(\Phi\cdot\varphi)=T_\eta\Phi\cdot T_\eta\varphi$.
\end{lemma}

\begin{proof}
If $\Phi,\varphi\in\cP$ then $T_\eta T_\xi\Phi=T_{\eta+\xi}\Phi$ and $T_\eta(\Phi\cdot\varphi)=T_\eta\Phi\cdot T_\eta\varphi$ by definition. Furthermore~\eqref{eq:chaos-decomposition-of-translation} implies $S(T_\eta\Phi)(0)=S\Phi(\eta)$. These identities instantly generalize to $\Phi\in\cG'$ and $\varphi\in\cG$ by approximation. Finally
\[ S(T_\eta\Phi)(\xi)=S(T_\xi T_\eta\Phi)(0)=S(T_{\eta+\xi}\Phi)(0)=S\Phi(\eta+\xi) \]
and the statement is proven.
\end{proof}

Since $L^p(\mu)\subset\cG'$ continuously for $p>1$ we have that $T_\eta$ is continuous from $L^p(\mu)$ to $\cG'$ and this definition is uniquely determined by the relation $T_\eta\varphi=\varphi(\cdot+\eta)$ for $\varphi\in\cP$. However, as mentioned before, it is not clear whether this relation extends to more general $\varphi$. In fact, $\varphi(\cdot+\eta)$ is not even well-defined for general $\eta\in L_\C^2(\R)$ and $\varphi\in L^p(\mu)$. Furthermore, one does not know whether $T_\eta\varphi\in L^q(\mu)$ for some $q\ge 1$.

The real case $\eta\in L^2(\R)$ is rather simple. The next statement is proven using that
\[ \frac{\rmd\mu(\cdot-\eta)}{\rmd\mu}=\exp\bigl(\la\cdot,\eta\ra-\la\eta,\eta\ra/2\bigr) \]
and Hölder's inequality.

\begin{lemma}\label{lem:realtranslationpointwise}
Let $\eta\in L^2(\R)$ and $1\le q<p<\infty$. Then $T_\eta$ maps $L^p(\mu)$ continuously to $L^q(\mu)$ and $T_\eta\varphi=\varphi(\cdot+\eta)$ for all $\varphi\in L^p(\mu)$.
\end{lemma}

For the complex case $\eta\in L_\C^2(\R)$ at least we get a closability result which is due to continuity of $L^p(\mu)\subset\cG'$ and Corollary~\ref{cor:translationoperator-continuous}.

\begin{corollary}
For $\eta\in L_\C^2(\R)$ and any $p,q\in(1,\infty)$ the operator $L^p(\mu)\supset\cP\ni\varphi\mapsto T_\eta\varphi\in L^q(\mu)$ is closable.
\end{corollary}

In the following proposition we illustrate that the identity $T_\nu\varphi=\varphi(\cdot+\eta)$ may fail to hold for complex $\eta$. Here we say that a measurable $\varphi\colon S'(\R)\to\C$ is an element of $\cG'$ if the mapping
\[ \cG\ni\psi\longmapsto\int_{S'(\R)}\varphi\cdot\psi\,\rmd\mu\in\C \]
is well-defined and continuous. Note that this implies that $\varphi$ is integrable, as $\indi_{S'(\R)}\in\cG$.

\begin{proposition}
For all $\eta\in L_\C^2(\R)$ with $\Imag\eta\not=0$ there exists some $\varphi\in L^\infty(\mu)$ which has an analytic extension to $S_\C'(\R)$ such that $\varphi(\cdot+\eta)\not\in\cG'$. In particular $T_\eta\varphi\not=\varphi(\cdot+\eta)$.
\end{proposition}

\begin{proof}
Pick $\xi\in S(\R)$ with $-3\lv\xi\rv^2\la\Imag\eta,\xi\ra>1/2$ and set $\varphi:=\exp\bigl(\rmi\la\cdot,\xi\ra^3\bigr)\in L^\infty(\mu)$. Obviously $\varphi$ has an analytic extension to $S_\C'(\R)$. Since $\lv\varphi(\cdot+\eta)\rv=\exp\bigl(\Real\rmi(\la\cdot,\xi\ra+\la\eta,\xi\ra)^3\bigr)$ and $\la\cdot,\xi\ra$ has a centered Gaussian distribution with variance $\lv\xi\rv^2$, it follows that
\[ \int_{S'(\R)}\lv\varphi(\omega+\eta)\rv\,\rmd\mu(\omega)=\int_\R\exp\Bigl(-3\lv\xi\rv^2\la\Imag\eta,\xi\ra x^2+c_1x+c_0\Bigr)\,\rmd\mu_\R(x)=\infty, \]
where $c_0,c_1\in\R$ depend on $\eta$ and $\xi$. Thus $\varphi(\cdot+\eta)$ is not integrable and hence not an element of $\cG'$.
\end{proof}

Below we show a method of how to obtain a representation for $T_\eta\varphi$ via the concept of infinite-dimensional holomorphy on the space $L_\C^2(\R)$, see~\cite{Mujica85} for extensive background on this topic.

The well-known identity theorem states that a holomorphic function defined on a connected open subset of $L_\C^2(\R)$ which vanishes on a non-empty open subset of $U$ already vanishes identically, see~\cite[Prop.~5.7]{Mujica85}. We prove a similar statement, namely that the same is true if the function vanishes on a non-empty open set of $L^2(\R)$.

\begin{proposition}\label{prop:identity-theorem-version}
Let $U$ be a connected open subset of $L_\C^2(\R)$ and $h\colon U\to\C$ be holomorphic. If $h$ vanishes on a non-empty open subset $V$ of $L^2(\R)$ which is completely contained in $U$, then $h$ vanishes identically.
\end{proposition}

\begin{proof}
By the identity theorem it suffices to show that $h$ vanishes on a non-empty open subset of $U$. Consider the Taylor expansion of $h$ at a fixed point $\xi_0\in V$: For each $m\in\N$ there exists a continuous $m$-homogeneous polynomial $P_m\colon L_\C^2(\R)\to\C$ such that $h(\xi)=\sum_{m=0}^\infty P_m(\xi-\xi_0)$ for all $\xi$ from a neighborhood of $\xi_0$ in $L_\C^2(\R)$. Knowing that $h\equiv 0$ on $V$ we get $P_m\equiv 0$ on $L^2(\R)$ for every $m\in\N$ by~\cite[Prop.~4.4]{Mujica85}. Since $P_m$ is uniquely determined by its values on $L^2(\R)$, it follows $P_m=0$ and thus $h=0$.
\end{proof}

\begin{theorem}\label{thm:translationrepresentation}
Let $\Phi\in\cG'$ and let $U$ be a connected open subset of $L_\C^2(\R)$, $\eta\in U$ and $(\Phi_\zeta)_{\zeta\in U}$ be a family in $\cG'$ such that
\begin{enumerate}
\item $U\cap L^2(\R)\not=\emptyset$.
\item $\E[\Phi_\zeta]=\E[T_\zeta\Phi]$ for all $\zeta\in U\cap L^2(\R)$.
\item $\E[\Phi_{\eta+\zeta}]=\E[T_\zeta\Phi_\eta]$ for all $\zeta\in(U-\eta)\cap L^2(\R)$.
\item The mapping $U\ni\zeta\mapsto\E[\Phi_\zeta]\in\C$ is holomorphic.
\end{enumerate}
Then $T_\eta\Phi=\Phi_\eta$.
\end{theorem}

\begin{proof}
We know $S\Phi(\zeta)=\E[T_\zeta\Phi]=\E[\Phi_\zeta]$ for all $\zeta\in U\cap L^2(\R)$. Since both sides of the equation are holomorphic in $\zeta\in U$, Proposition~\ref{prop:identity-theorem-version} gives equality for all $\zeta\in U$. Then
\[ S(T_\eta\Phi)(\zeta)=S\Phi(\zeta+\eta)=\E[\Phi_{\eta+\zeta}]=\E[T_\zeta\Phi_\eta]=S\Phi_\eta(\zeta). \]
For all $\zeta\in(U-\eta)\cap L^2(\R)$, which is a non-empty open set in $L^2(\R)$. Another application of Proposition~\ref{prop:identity-theorem-version} yields $S(T_\eta\Phi)=S\Phi_\eta$ and thus $T_\eta\Phi=\Phi_\eta$.
\end{proof}

\begin{corollary}\label{cor:translationrepresentation}
Let $(\Phi_\zeta)_{\zeta\in L_\C^2(\R)}$ be a family in $\cG'$ such that
\begin{enumerate}
\item $\E[\Phi_{\eta+\zeta}]=\E[T_\zeta\Phi_\eta]$ for all $\eta\in L_\C^2(\R)$ and $\zeta\in L^2(\R)$.
\item The mapping $L_\C^2(\R)\ni\zeta\mapsto\E[\Phi_\zeta]\in\C$ is holomorphic.
\end{enumerate}
Then $T_\eta\Phi_\zeta=\Phi_{\zeta+\eta}$ for all $\eta,\zeta\in L_\C^2(\R)$.
\end{corollary}

\begin{proof}
The conditions of Theorem~\ref{thm:translationrepresentation} are trivially fulfilled for $\Phi=\Phi_0$, $U=L_\C^2(\R)$, the family $(\Phi_\zeta)_{\zeta\in L_\C^2(\R)}$ and an arbitrary $\eta\in L_\C^2(\R)$. Hence $T_\eta\Phi_0=\Phi_\eta$ for all $\eta\in L_\C^2(\R)$. It follows $T_\eta\Phi_\zeta=T_\eta T_\zeta\Phi_0=T_{\eta+\zeta}\Phi_0=\Phi_{\eta+\zeta}$ for all $\eta,\zeta\in L_\C^2(\R)$.
\end{proof}

\begin{example}\label{ex:translation-pointwise}
We present some applications of Corollary~\ref{cor:translationrepresentation}. The first two examples can also be verified by computing the chaos decomposition directly.
\begin{itemize}
\item If $\Phi=P(\la\cdot,\xi_1\ra,\dotsc,\la\cdot,\xi_d\ra)\in\cP_\xi$ for some polynomial $P\colon\C^d\to\C$ and $\xi\in L^2(\R)^d$, then $T_\eta\Phi=P(\la\cdot,\xi_1\ra+\la\eta,\xi_1\ra,\dotsc,\la\cdot,\xi_d\ra+\la\eta,\xi_d\ra)$ for all $\eta\in L_\C^2(\R)$. In particular $T_\eta\Phi\in\cP_\xi$, thus $\cP_\xi$ is left invariant by $T_\eta$.
\item Fix $\xi\in L_\C^2(\R)$ and set $\Phi_\zeta:=\exp(\la\cdot,\xi\ra+\la\zeta,\xi\ra)\in\cG\subset\cG'$ for $\zeta\in L_\C^2(\R)$. From Lemma~\ref{lem:realtranslationpointwise} we know $T_\zeta\Phi_\eta=\Phi_{\eta+\zeta}$ for all $\zeta\in L^2(\R)$. The mapping
\[ L_\C^2(\R)\ni\zeta\longmapsto\E[\Phi_\zeta]=\exp\Bigl(\frac{1}{2}\la\xi,\xi\ra+\la\zeta,\xi\ra\Bigr)\in\C \]
is obviously holomorphic. Hence $T_\eta\Phi_\zeta=\Phi_{\zeta+\eta}$ for all $\eta,\zeta\in L_\C^2(\R)$.
\item
Let $V$ be a polynomial from the class of potentials under consideration and fix $x\in\C$, $d\in\N$, $c_1,\dotsc,c_d>0$ and $\xi_1,\dotsc,\xi_d\in L^2(\R)$. Consider
\[ \Phi_\zeta:=\exp\biggl(-\rmi\sum_{k=1}^dc_kV\bigl(x+\sqrt{\rmi}\la\cdot,\xi_k\ra+\sqrt{\rmi}\la\zeta,\xi_k\ra\bigr)\biggr) \]
for $\zeta\in L_\C^2(\R)$. As in the proof of Lemma~\ref{lem:boundedness-of-exponential-term} it is easy to see that $\Phi_\zeta\in L^2(\mu)$. Furthermore $T_\zeta\Phi_\eta=\Phi_{\zeta+\eta}$ for $\eta\in L_\C^2(\R)$ and $\zeta\in L^2(\R)$ by Lemma~\ref{lem:realtranslationpointwise}. Note that $\E[\Phi_\zeta]$ is locally bounded in $\zeta\in L_\C^2(\R)$, which can be seen as in the proof of Lemma~\ref{lem:boundedness-of-exponential-term} as well. To apply Corollary~\ref{cor:translationrepresentation} it is left to show holomorphy of $\E[\Phi_\zeta]$, for which by~\cite[Prop.~8.6,~Thm.~8.7]{Mujica85} it suffices to show that this mapping is G-holomorphic, i.e.\ for all $\zeta_0,\zeta_1\in L_\C^2(\R)$ the mapping $\C\ni z\mapsto\E[\Phi_{\zeta_0+z\zeta_1}]\in\C$ is holomorphic. This can be proven by applying the theorems of Fubini and Morera: Let $\gamma$ be a closed curve in $\C$. Note that $\C\ni z\mapsto\Phi_{\zeta_0+z\zeta_1}\in\C$ is holomorphic almost surely, since $V$ is a polynomial. Hence Morera's theorem applied to this function yields
\[ \int_\gamma\E[\Phi_{\zeta_0+z\zeta_1}]\,\rmd z=\E\biggl[\int_\gamma\Phi_{\zeta_0+z\zeta_1}\,\rmd z\biggr]=\E[0]=0, \]
where exchange of expectation and integration is possible due to locally boundedness of $z\mapsto\E[\Phi_{\zeta_0+z\zeta_1}]$ and compactness of the range of the path in $\C$. Then another application of Morera's theorem shows that $\E[\Phi_{\zeta_0+z\zeta_1}]$ is in fact holomorphic in $z\in\C$, so Corollary~\ref{cor:translationrepresentation} shows $T_\eta\Phi_\zeta=\Phi_{\zeta+\eta}$ for all $\eta,\zeta\in L_\C^2(\R)$.
\end{itemize}
\end{example}

\end{subsubsection}

\begin{subsubsection}{Projection Operator}

Let $\eta\in L^2(\R)$ with $\lv\eta\rv=1$. The projection operator $P_\eta$, introduced in~\cite{Westerkamp1995}, is a continuous linear operator $P_\eta\colon\cG\to\cG$, which aims to remove the dependency on the monomial $\la\cdot,\eta\ra$ from a random variable. Again, the domain of definition $\cG$ is too small for our purposes. In~\cite{GRS14} an extension of $P_\eta$ to $\overline{\cP_\xi}^{L^p(\mu)}$ was presented, where $\xi\in L^2(\R)^d$ is such that $\eta\not\in\spann\{\xi_1,\dotsc,\xi_d\}$, see Proposition~\ref{prop:GRS14}. This was sufficient for realizing the product with Donsker's delta in the special case $a=0$, i.e.\ multiplication with $\delta(\la\cdot,\eta\ra)$. To generalize to the case $a\not=0$, we need a stronger result, namely a continuous extension of $P_\eta$ to the closure of $\cP_\xi$ in $\cG_q$, where $q$ is negative. A sufficient estimate is proven in Proposition~\ref{prop:projcontonnegative}.

We repeat the definition of $P_\eta$ on $\cG$: Let $P_{\bot,\eta}\colon L^2(\R)\to L^2(\R)$ denote the orthogonal projection onto $\{\eta\}^{\bot}$, so $P_{\bot,\eta}\xi=\xi-\la\xi,\eta\ra\eta$ for $\xi\in L^2(\R)$. We also consider the complexification of this operator and denote it by the same symbol. Then $P_\eta\colon\cP\to\cG$ is defined by
\begin{equation}\label{eq:polyprojection}
P_\eta\varphi:=\sum_{n=0}^N\biggl\la\wickc,P_{\bot,\eta}^\otn\sum_{k=0}^{\lfloor (N-n)/2\rfloor}\frac{(n+2k)!(-1)^k}{n!k!2^k}\,\eta^{\ot 2k}\wot_{2k}\varphi^{(n+2k)}\biggr\ra
\end{equation}
for $\varphi\in\cP$ as in~\eqref{eq:smooth-polynomial-chaos-decomposition}. It was shown in~\cite[Thm.~71]{Westerkamp1995} that $P_\eta$ extends uniquely to a continuous linear operator $P_\eta\colon\cG\to\cG$.

In the special case $\eta\in S(\R)$ the operator $P_{\bot,\eta}$ can be extended to $P_{\bot,\eta}\colon S'(\R)\to S'(\R)$ via $P_{\bot,\eta}\omega:=\omega-\la\omega,\eta\ra\eta$ for $\omega\in S'(\R)$. In that case, for $\varphi\in\cP$ it holds $P_\eta\varphi=\varphi\circ P_{\bot,\eta}=\varphi(\cdot-\la\cdot,\eta\ra\eta)$, i.e.\ the projection operator projects the argument $\omega$ of the random variable $\varphi$ onto the space $\{\omega\in S'(\R):\la\omega,\eta\ra=0\}$. This is the original motivation for this definition of the projection operator. Note however that $\la\cdot,\eta\ra$ is a standard Gaussian random variable, so $\{\la\cdot,\eta\ra=0\}$ is a $\mu$-nullset and it is the special property of a smooth polynomial to have a unique pointwisely defined continuous version with respect to the Gaussian measure $\mu$ which allows evaluation on that nullset.

The following is a result from~\cite{GRS14}.

\begin{proposition}\label{prop:GRS14}
Let $\eta,\xi_1,\dotsc,\xi_d\in L^2(\R)$ with $\lv\eta\rv=1$ such that $\eta\not\in\spann\{\xi_1,\dotsc,\xi_d\}$. Then for each $p\in[1,\infty)$ the operator $P_\eta$ extends uniquely to a bounded linear operator from $\overline{\cP_\xi}^{L^p(\mu)}$ to $L^p(\mu)$, where $\xi=(\xi_1,\dotsc,\xi_d)$. For $f(\la\cdot,\xi_1\ra,\dotsc,\la\cdot,\xi_d\ra)\in\overline{\cP_\xi}^{L^p(\mu)}$ as in Lemma~\ref{lem:finite-dimensional-Lp-isometry} this extension is given by
\[  P_{\eta}f(\la\cdot,\xi_1\ra,\dotsc,\la\cdot,\xi_d\ra)=f(\la\cdot,P_{\bot,\eta}\xi_1\ra,\dotsc,\la\cdot,P_{\bot,\eta}\xi_d\ra).\]
\end{proposition}

It is clear that the condition $\eta\not\in\spann\{\xi_1,\dotsc,\xi_d\}$ is necessary: Consider an arbitrary $f\in L^2(\R,\mu_1;\C)$, where $\mu_1$ denotes the standard Gaussian measure on $\R$. Then $f(\la\cdot,\eta\ra)\in\overline{\cP_{(\eta)}}^{L^2(\mu)}$ and $f(\la\cdot,P_{\bot,\eta}\eta\ra)=f(0)$ corresponds to evaluation of $f$ at $0$, which is not well-defined.

This proposition is useful to obtain a representation of $P_\eta\varphi$ when $\varphi\in\overline{\cP_\xi}^{L^p(\mu)}$. It also suffices to prove Theorem~\ref{thm:wickformula-extended} below in the special case $a=0$, see~\cite{GRS14}. As mentioned before, we need continuity of $P_\eta$ with respect $\cG_q$ for some $q<0$ to cover the case $a\not=0$. The following estimate is used in the corresponding proof.

\begin{lemma}\label{lem:restricted-kernel-estimate}
Let $\eta,\xi_1,\dotsc,\xi_d\in L^2(\R)$ with $\lv\eta\rv=1$ and $\eta\not\in E:=\spann_\C\{\xi_1,\dotsc,\xi_d\}$. Then there exists $c<1$ with
\[ \bigl\lv\eta^{\ot 2k}\ot_{2k}\varphi^{(n+2k)}\bigr\rv\le c^{2k}\bigl\lv\varphi^{(n+2k)}\bigr\rv \]
for all $n,k\in\N$ and $\varphi^{(n+2k)}\in E^{\ot n+2k}$, where $E^{\ot n+2k}$ is the $(n+2k)$-fold Hilbert space tensor product of $E$.
\end{lemma}

\begin{proof}
Define the linear operator $Q\colon E\to\C$ by $Q\xi:=\la\eta,\xi\ra$ for $\xi\in E$. Then
\[ \eta^{\ot 2k}\ot_{2k}\varphi^{(n+2k)}=(Q^{\ot 2k}\ot\Id^\otn)\varphi^{(n+2k)} \]
for $\varphi^{(n+2k)}\in E^{\ot n+2k}$.
Since $\eta\not\in E$ we have $c:=\lV Q\rV_{L(E,\C)}<1$, where $\lV Q\rV_{L(E,\C)}$ denotes the operator norm of $Q$. Thus for the operator norm of $Q^{\ot 2k}\ot\Id^\otn$ it holds $\lV Q^{\ot 2k}\ot\Id^\otn\rV_{L(E^{\ot n+2k},\C)}\le c^{2k}$ which proves the assertion.
\end{proof}

\begin{proposition}\label{prop:projcontonnegative}
Let $\eta,\xi_1,\dotsc,\xi_d\in L^2(\R)$ with $\lv\eta\rv=1$ such that $\eta\not\in\spann\{\xi_1,\dotsc,\xi_d\}$ and set $\xi=(\xi_1,\dotsc,\xi_d)$. Then there exist $q,r<0$ and $C_2=C_2(q,r,\eta,\xi)\in\R$ such that $\lV P_\eta\varphi\rV_r\le C_2\lV\varphi\rV_q$ for all $\varphi\in\cP_\xi$. In particular, $P_\eta$ extends uniquely to a bounded linear operator from the closure of $\cP_\xi$ in $\cG_q$ to $\cG_r$.
\end{proposition}

\begin{proof}
By Lemma~\ref{lem:restricted-kernel-estimate} there exists $\eps>0$ such that
\[ \bigl\lv\eta^{\ot 2k}\wot_{2k}\varphi^{(n+2k)}\bigr\rv\le 2^{-k\eps}\cdot\bigl\lv\varphi^{(n+2k)}\bigr\rv \]
for $\varphi\in\cP_\xi$ as in~\eqref{eq:smooth-polynomial-chaos-decomposition} and all $n,k\in\N$. Choose $q,r<0$ such that $-\eps<q<0$ and $2^{-(q+\eps)}+2^{r-q}<1$. Using $(2k)!\le k!k!2^{2k}$ we compute that for the $n$-th kernel $\tilde{\varphi}^{(n)}$ of $P_\eta\varphi$ as in~\eqref{eq:polyprojection} it holds
\begin{align*}
\bigl\lv\tilde{\varphi}^{(n)}\bigr\rv^2
 & \le \biggl(\sum_{k=0}^{\lfloor (N-n)/2\rfloor}\frac{(n+2k)!}{n!k!2^k}\bigl\lv\eta^{\ot 2k}\wot_{2k}\varphi^{(n+2k)}\bigr\rv\biggr)^2\\
 & \le \biggl(\sum_{k=0}^{\lfloor (N-n)/2\rfloor}(n+2k)!2^{q(n+2k)}\bigl\lv\varphi^{(n+2k)}\bigr\rv^2\biggr)\cdot\biggl(\sum_{k=0}^{\lfloor (N-n)/2\rfloor}\frac{(n+2k)!}{n!n!k!k!2^{2k}}2^{-q(n+2k)}2^{-2k\eps}\biggr)\\
 & \le \lV\varphi\rV_q^2\cdot\frac{1}{n!}2^{-qn}\sum_{k=0}^{N-n}\binom{n+k}{k}2^{-(q+\eps)k}
\end{align*}
which yields
\begin{align*}
\lV P_\eta\varphi\rV_r^2=\sum_{n=0}^Nn!2^{rn}\bigl\lv\tilde{\varphi}^{(n)}\bigr\rv^2
 & \le \lV\varphi\rV_q^2\cdot\sum_{n=0}^N\sum_{k=0}^{N-n}\binom{n+k}{k}2^{(r-q)n}2^{-(q+\eps)k}\\
 & = \lV\varphi\rV_q^2\cdot\sum_{n=0}^N\sum_{k=0}^n\binom{n}{k}2^{(r-q)(n-k)}2^{-(q+\eps)k}\\
 & = \lV\varphi\rV_q^2\cdot\sum_{n=0}^N\bigl(2^{r-q}+2^{-(q+\eps)}\bigr)^n\\
 & \le\bigl(1-2^{r-q}-2^{-(q+\eps)}\bigr)^{-1}\cdot\lV\varphi\rV_q^2
\end{align*}
and the statement is proven.
\end{proof}

\end{subsubsection}

\begin{subsubsection}{Pointwise Multiplication with Donsker's Delta}

It was proven in~\cite[Thm.~4.24]{Vogel2010} that for $\varphi\in\cG$, $a\in\C$ and $\eta\in L^2(\R)\setminus\{0\}$ it holds
\[ \delta(\la\cdot,\eta\ra-a)\cdot\varphi=\delta(\la\cdot,\eta\ra-a)\diamond P_{\frac{\eta}{\lv\eta\rv}}T_{\frac{a\eta}{\lv\eta\rv^2}}\varphi. \]
Using homogeneity of Donskers's delta of degree $-1$ it directly follows that
\begin{equation}\label{eq:scaledwickformula}
\delta(z\la\cdot,\eta\ra-a)\cdot\varphi=\delta(z\la\cdot,\eta\ra-a)\diamond P_{\frac{\eta}{\lv\eta\rv}}T_{\frac{a\eta}{z\lv\eta\rv^2}}\varphi
\end{equation}
for $z\in\C$ with $\Real z>0$. This formula implies that extensions of the composition $P_{\frac{\eta}{\lv\eta\rv}}T_{\frac{a\eta}{\lv\eta\rv^2}}$ give extensions of the pointwise multiplication with Donsker's delta. Combining our previous continuity results gives the following extension for the composition.

\begin{proposition}\label{prop:compositioncont}
Let $\eta,\xi_1,\dotsc,\xi_d\in L^2(\R)$ with $\lv\eta\rv=1$ such that $\eta\not\in\spann\{\xi_1,\dotsc,\xi_d\}$ and let $\nu\in L_\C^2(\R)$ be arbitrary. Then $P_\eta T_\nu$ is continuous from $\overline{\cP_\xi}^{L^2(\mu)}$ to $\cG'$, where $\xi=(\xi_1,\dotsc,\xi_d)$.
\end{proposition}

\begin{proof}
Let $q,r<0$ be as in Proposition~\ref{prop:projcontonnegative}. By Example~\ref{ex:translation-pointwise} it holds $T_\nu(\cP_\xi)\subset\cP_\xi$, thus
\[ \lV P_\eta T_\nu\varphi\rV_r\le C_2(q,r,\eta,\xi)\lV T_\nu\varphi\rV_q\le C_2(q,r,\eta,\xi)C_1(0,q,\lv\nu\rv)\lV\varphi\rV_{L^2(\mu)} \]
for all $\varphi\in\cP_\xi$, where $C_1$ and $C_2$ are as in Propositions~\ref{prop:translationoperator-continuous} and~\ref{prop:projcontonnegative}, respectively.
\end{proof}

The following theorem shows that the pointwise multiplication with Donsker's delta can be defined for a subclass of functions from $L^2(\mu)$ by extending~\eqref{eq:scaledwickformula} abstractly. An explicit representation of the pointwise product can be obtained in combination with Theorem~\ref{thm:translationrepresentation} and Proposition~\ref{prop:GRS14}.

\begin{theorem}\label{thm:wickformula-extended}
Let $a,z\in\C$ with $\Real z>0$ and $\eta,\xi_1,\dotsc,\xi_d\in L^2(\R)$ such that $\eta\not\in\spann\{\xi_1,\dotsc,\xi_d\}$. Then the linear operator
\[ \cP_\xi\ni\varphi\mapsto\delta(z\la\cdot,\eta\ra-a)\cdot\varphi\in\cG' \]
has a unique continuous extension to $\overline{\cP_\xi}^{L^2(\mu)}$ which is given as in~\eqref{eq:scaledwickformula} for $\varphi\in\overline{\cP_\xi}^{L^2(\mu)}$.
\end{theorem}

\begin{proof}
We know that~\eqref{eq:scaledwickformula} is valid for $\varphi\in\cP_\xi$. Proposition~\ref{prop:compositioncont} and continuity of the Wick product from $\cG'\times\cG'$ to $\cG'$ finish the proof.
\end{proof}

The application we have in mind is to define~\eqref{eq:riemann-approximation-limit}, so we consider functions which depend on Brownian motion.

\begin{example}\label{ex:brownianmotion-to-brownianbridge}
Let $0<t_1<\dotsb<t_d<t$ and $x,y\in\C$. Define $\eta:=\indi_{[0,t)}/\sqrt{t}$ and $\nu:=\frac{y-x}{\sqrt{\rmi}t}\indi_{[0,t)}$. Assume $f\colon\C^d\to\C$ is measurable such that
\[ \varphi:=f\bigl(x+\sqrt{\rmi}B_{t_1},\dotsc,x+\sqrt{\rmi}B_{t_d}\bigr)\in L^2(\mu). \]
Suppose furthermore that we have the equality
\[ T_\nu\varphi=f\bigl(x+\sqrt{\rmi}B_{t_1}+\sqrt{\rmi}\la\nu,\indi_{[0,t_1)}\ra,\dotsc,x+\sqrt{\rmi}B_{t_d}+\sqrt{\rmi}\la\nu,\indi_{[0,t_d)}\ra\bigr), \]
which can be verified via Theorem~\ref{thm:translationrepresentation}, and that $T_\nu\varphi\in L^2(\mu)$ holds. Then Proposition~\ref{prop:GRS14} shows
\[ P_\eta T_\nu\varphi=f\biggl(x+\frac{t_1}{t}(y-x)+\sqrt{\rmi}B_{t_1}-\sqrt{\rmi}\frac{t_1}{t}B_t,\dotsc,x+\frac{t_d}{t}(y-x)+\sqrt{\rmi}B_{t_d}-\sqrt{\rmi}\frac{t_d}{t}B_t\biggr), \]
so the operator $P_\eta T_\nu$ transforms the finitely many samples of the scaled Brownian motion $(x+\sqrt{\rmi}B_r)_{r\in[0,t]}$ starting in $x$ into corresponding samples of the scaled Brownian bridge $\bigl(x+\frac{r}{t}(y-x)+\sqrt{\rmi}B_r-\sqrt{\rmi}\frac{r}{t}B_t\bigr)_{r\in[0,t]}$ starting in $x$ and ending in $y$. Thus we have found the representation
\[ \delta(\sqrt{\rmi}B_t-(y-x))\cdot\varphi=\delta(\sqrt{\rmi}B_t-(y-x))\diamond P_\eta T_\nu\varphi \]
for the pointwise product with Donsker's delta in Theorem~\ref{thm:wickformula-extended}.
\end{example}

\end{subsubsection}

\end{subsection}

\end{section}

\begin{section}{The Fundamental Solution to the Schrödinger Equation}

\begin{subsection}{Construction of the Fundamental Solution}

We construct the integrand of $K_V(t,x;y)$ as in~\eqref{eq:def-KV} by the approximation procedure~\eqref{eq:riemann-approximation-limit}, where we interpret $\delta_y\bigl(x+\sqrt{\rmi}B_t\bigr)$ as Donsker's delta $\delta(\sqrt{\rmi}\la\cdot,\indi_{[0,t)}\ra-(y-x))$. More precisely, we use Theorem~\ref{thm:wickformula-extended} to see that
\begin{equation}\label{eq:riemannapprox}
\delta_y\bigl(x+\sqrt{\rmi}B_t\bigr)\cdot\exp\biggl(-\rmi\sum_{[r,s]\in Q}(s-r)V\bigl(x+\sqrt{\rmi}B_r\bigr)\biggr)
\end{equation}
is well-defined in $\cG'$ for $t>0$, $x,y\in\C$ and every partition $Q$ of the interval $[0,t]$, and we show that~\eqref{eq:riemannapprox} converges in $\cG'$ as the mesh of $Q$ approaches zero. A similar method was used in~\cite{Vogel2010,GSV12}, where the definition of~\eqref{eq:riemannapprox} and~\eqref{eq:riemann-approximation-limit} was performed via an integration procedure and relied on the choice of equidistant partitions of the interval $[0,t]$. However, the definition of~\eqref{eq:riemannapprox} is unique in the sense of Theorem~\ref{thm:wickformula-extended}, and~\eqref{eq:riemann-approximation-limit} is independent of the particular choice of the sequence $(Q_n)_{n\in\N}$. With the help of Example~\ref{ex:brownianmotion-to-brownianbridge} we show that $K_V(t,x;y)$ admits the probabilistic representation~\eqref{eq:KV-representation-intro}.

To be precise in what follows, we say that $Q$ is a partition of $[0,t]$ if there are some $d\in\N$ and $0=t_0<t_1<\dotsb<t_{d-1}<t_d=t$ such that $Q=\{[t_{i-1},t_i]:i=1,\dots,d\}$. The mesh of such $Q$ is defined as $\max_{[r,s]\in Q}(s-r)$.

Let $(Q_n)_{n\in\N}$ be an arbitrary sequence of partitions of $[0,t]$ whose mesh converges to zero and define the random variables
\[ \varphi(t,x):=\exp\biggl(-\rmi\int_0^tV\bigl(x+\sqrt{\rmi}B_r\bigr)\,\rmd r\biggr) \]
and
\[ \varphi_n(t,x):=\exp\biggl(-\rmi\sum_{[r,s]\in Q_n}(s-r)V\bigl(x+\sqrt{\rmi}B_r\bigr)\biggr) \]
for $n\in\N$. We summarize the results of this section in the following theorem.

\begin{theorem}\label{thm:KV-construction}
For $t>0$ and $x,y\in\C$ we have that
\begin{equation}\label{eq:riemann-approximation-rewritten}
\delta_y\bigl(x+\sqrt{\rmi}B_t\bigr)\cdot\varphi(t,x):=\lim_{n\to\infty}\delta_y\bigl(x+\sqrt{\rmi}B_t\bigr)\cdot\varphi_n(t,x)
\end{equation}
exists in $\cG'$ independently of the sequence $(Q_n)_{n\in\N}$ of partitions of $[0,t]$. It holds
\begin{multline}\label{eq:wick-product-representation-of-KV-integrand}
\delta_y\bigl(x+\sqrt{\rmi}B_t\bigr)\cdot\varphi(t,x)=\delta_y\bigl(x+\sqrt{\rmi}B_t\bigr)\\
\diamond\exp\biggl(-\rmi\int_0^tV\Bigl(x+\frac{r}{t}(y-x)+\sqrt{\rmi}B_r-\frac{r}{t}\sqrt{\rmi}B_t\Bigr)\,\rmd r\biggr).
\end{multline}
For its expectation $K_V(t,x;y):=\E[\delta_y(x+\sqrt{\rmi}B_t)\cdot\varphi(t,x)]$ we have
\begin{equation}\label{eq:KV-representation}
K_V(t,x;y)=K_0(t,x;y)\cdot\E\biggl[\exp\biggl(-\rmi\int_0^tV\Bigl(y+\frac{r}{t}(x-y)+\sqrt{\rmi}B_r-\frac{r}{t}\sqrt{\rmi}B_t\Bigr)\,\rmd r\biggr)\biggr].
\end{equation}
In particular $K_V(t,x;y)=K_V(t,y;x)$.
\end{theorem}

Here~\eqref{eq:riemann-approximation-rewritten} and~\eqref{eq:KV-representation} are just reformulations of~\eqref{eq:riemann-approximation-limit} and~\eqref{eq:KV-representation-intro}. Note that in~\eqref{eq:wick-product-representation-of-KV-integrand} and~\eqref{eq:KV-representation} we have switched the roles of $x$ and $y$ in the exponential term.

The definition in~\eqref{eq:riemann-approximation-rewritten} is reasonable: Obviously $\lim_{n\to\infty}\varphi_n(t,x)=\varphi(t,x)$ almost surely and independently of the choice of $(Q_n)_{n\in\N}$, since Brownian motion has continuous paths. Furthermore, similar as in the proof of Lemma~\ref{lem:boundedness-of-exponential-term}, we see $\sup_{n\in\N}\lV\varphi_n(t,x)\rV_{L^\infty(\mu)}<\infty$, hence by Lebesgue's dominated convergence theorem it also holds $\lim_{n\to\infty}\varphi_n(t,x)=\varphi(t,x)$ in $L^2(\mu)$ and in particular in $\cG'$.

We show existence of~\eqref{eq:riemann-approximation-rewritten} in $\cG'$: Let $\eta:=\indi_{[0,t)}/\sqrt{t}$, $\nu:=\frac{y-x}{t\sqrt{\rmi}}\indi_{[0,t)}$, fix $n\in\N$ and set $\xi:=(\indi_{[0,r)})_{[r,s]\in Q_n}\in L^2(\R)^{\lv Q_n\rv}$. Using  Lemma~\ref{lem:boundedness-of-exponential-term} and Examples~\ref{ex:translation-pointwise} and~\ref{ex:brownianmotion-to-brownianbridge} we see
\[ \tilde{\varphi}_n(t,x;y):=P_\eta T_\nu\varphi_n(t,x)=\exp\biggl(-\rmi\sum_{[r,s]\in Q_n}(s-r)V\Bigl(x+\frac{r}{t}(y-x)+\sqrt{\rmi}B_r-\frac{r}{t}\sqrt{\rmi}B_t\Bigr)\biggr). \]
and $\tilde{\varphi}_n(t,x;y)\in L^2(\mu)$. Note that $\varphi_n(t,x)\in\overline{\cP_\xi}^{L^2(\mu)}$. Applying Theorem~\ref{thm:wickformula-extended} yields
\[ \delta_y\bigl(x+\sqrt{\rmi}B_t\bigr)\cdot\varphi_n(t,x)=\delta_y\bigl(x+\sqrt{\rmi}B_t\bigr)\diamond\tilde{\varphi}_n(t,x;y). \]
As before, the proof of Lemma~\ref{lem:boundedness-of-exponential-term} and Lebesgue's dominated convergence theorem show that $\tilde{\varphi}_n(t,x;y)$ converges to
\[ \tilde{\varphi}(t,x;y):=\exp\biggl(-\rmi\int_0^tV\Bigl(x+\frac{r}{t}(y-x)+\sqrt{\rmi}B_r-\frac{r}{t}\sqrt{\rmi}B_t\Bigr)\,\rmd r\biggr) \]
pointwisely and in $L^2(\mu)$ as $n\to\infty$, so in particular we have convergence in $\cG'$. This shows that~\eqref{eq:riemann-approximation-rewritten} exists independently of $(Q_n)_{n\in\N}$ and is given by~\eqref{eq:wick-product-representation-of-KV-integrand}. In particular we have rigorously defined $K_V(t,x;y)$ as in~\eqref{eq:def-KV} for $x,y\in\C$ and $t>0$ and it holds
\[ K_V(t,x;y)=\E\bigl[\delta_y\bigl(x+\sqrt{\rmi}B_t\bigr)\bigr]\cdot\E\biggl[\exp\biggl(-\rmi\int_0^tV\Bigl(x+\frac{r}{t}(y-x)+\sqrt{\rmi}B_r-\frac{r}{t}\sqrt{\rmi}B_t\Bigr)\,\rmd r\biggr)\biggr]. \]
This object was also derived in~\cite{Westerkamp1995,Vogel2010,GSV12}. Here we have used the results of Section~\ref{sec:whitenoiseanalysis} to establish uniqueness of the definition. It is left to show that~\eqref{eq:KV-representation} is valid. Since the Riemann integral is invariant with respect to reversion of the integration parameter, and the reverse of a Brownian bridge, i.e.\ the process $(B_{t-r}-\frac{t-r}{t}B_t)_{r\in[0,t]}$, is again a Brownian bridge, it holds
\begin{multline*}
\int_0^tV\Bigl(x+\frac{r}{t}(y-x)+\sqrt{\rmi}B_r-\frac{r}{t}\sqrt{\rmi}B_t\Bigr)\,\rmd r\\
= \int_0^tV\Bigl(x+\frac{t-r}{t}(y-x)+\sqrt{\rmi}B_{t-r}-\frac{t-r}{t}\sqrt{\rmi}B_t\Bigr)\,\rmd r\\
\stackrel{\mathcal{L}}{=} \int_0^tV\Bigl(y+\frac{r}{t}(x-y)+\sqrt{\rmi}B_r-\frac{r}{t}\sqrt{\rmi}B_t\Bigr)\,\rmd r.
\end{multline*}
Together with $K_0(t,x;y)=\E[\delta_y(x+\sqrt{\rmi}B_t)]$ this shows~\eqref{eq:KV-representation} and $K_V(t,x;y)=K_V(t,y;x)$. To shorten some notation we define
\begin{equation}\label{eq:def-psi}
\psi(t,x;y):=\E\biggl[\exp\biggl(-\rmi\int_0^tV\Bigl(y+\frac{r}{t}(x-y)+\sqrt{\rmi}B_r-\frac{r}{t}\sqrt{\rmi}B_t\Bigr)\,\rmd r\biggr)\biggr].
\end{equation}
so that $K_V(t,x;y)=K_0(t,x;y)\cdot\psi(t,x;y)$.

It should be emphasized that even though we used white noise distribution theory to define the integrand~\eqref{eq:riemann-approximation-rewritten} of $K_V(t,x;y)$, its representation~\eqref{eq:KV-representation} is a purely probabilistic expression which is well-defined outside of the white noise framework. Moreover~\eqref{eq:KV-representation} may be well-defined for a larger class of potentials.

\end{subsection}

\begin{subsection}{Differentiability of the Fundamental Solution}

In this section we compute the partial derivatives of $K_V$ and show that $K_V$ is continuously differentiable on $(0,\infty)\times\C^2$. Since $K_0$ is easy to differentiate it suffices to have a look at $\psi$ as in~\eqref{eq:def-psi}.

We start with computing the right-sided derivative of $\psi$ with respect to $t>0$.

\begin{proposition}\label{prop:psi-derivative}
For $x,y\in\C$ the right-sided derivative of $\psi$ with respect to $t>0$ can be computed as
\[ \partial_t^+\psi(t,x;y)=\Bigl(-\rmi V(x)-\frac{x-y}{t}\partial_x+\frac{\rmi}{2}\Delta\Bigr)\psi(t,x;y). \]
\end{proposition}

This proposition is proven by summing up the following three Lemmas. We follow the proof of~\cite[Thm.~5.20]{Vogel2010} closely and add some more details.

\begin{lemma}\label{lem:psi-derivative-1}
For $t>0$ and $x,y\in\C$ it holds
\begin{multline*}
\lim_{h\to 0}\frac{1}{h}\E\biggl[\exp\biggl(-\rmi\int_0^{t+h}V\Bigl(y+\frac{r}{t+h}(x-y)+\sqrt{\rmi}B_r-\frac{r}{t+h}\sqrt{\rmi}B_{t+h}\Bigr)\,\rmd r\biggr)\\
-\exp\biggl(-\rmi\int_0^tV\Bigl(y+\frac{r}{t+h}(x-y)+\sqrt{\rmi}B_r-\frac{r}{t+h}\sqrt{\rmi}B_{t+h}\Bigr)\,\rmd r\biggr)\biggr]\\
=  -\rmi V(x)\E\biggl[\exp\biggl(-\rmi\int_0^tV\Bigl(y+\frac{r}{t}(x-y)+\sqrt{\rmi}B_r-\frac{r}{t}\sqrt{\rmi}B_t\Bigr)\,\rmd r\biggr)\biggr].
\end{multline*}
\end{lemma}

\begin{proof}
Choose $\delta\in(0,t)$. For $\lv h\rv\le\delta$ and $u\ge 0$ define the complex-valued random variable
\[ F_h(u):=\exp\biggl(-\rmi\int_0^uV\Bigl(y+\frac{r}{t+h}(x-y)+\sqrt{\rmi}B_r-\frac{r}{t+h}\sqrt{\rmi}B_{t+h}\Bigr)\,\rmd r\biggr). \]
Then the assertion follows if we can show
\begin{equation}\label{eq:psi-derivative-1-reformulation}
\lim_{h\to 0}\frac{1}{h}\bigl(F_h(t+h)-F_h(t)\bigr)=-\rmi V(x)F_0(t)
\end{equation}
in $L^1(\mu)$. The derivative of $F_h$ is given by
\begin{equation}\label{eq:psi-derivative-1-derivative}
F_h'(u)= -\rmi V\Bigl(y+\frac{u}{t+h}(x-y)+\sqrt{\rmi}B_u-\frac{u}{t+h}\sqrt{\rmi}B_{t+h}\Bigr)F_h(u)
\end{equation}
and continuous in $u\ge 0$. Thus
\[ \frac{1}{h}\bigl(F_h(t+h)-F_h(t)\bigr)=\int_0^1F_h'(t+sh)\,\rmd s \]
by the fundamental theorem of calculus. By Lemmas~\ref{lem:boundedness-of-exponential-term} and~\ref{lem:uniform-polynomial-bound} we have that
\[ \sup_{\lv h\rv\le\delta}\sup_{s\in[0,1]}\lv F_h'(t+sh)\rv\le\rme^{cT}\cdot P(\lV B\rV_T)\in L^1(\mu) \]
for some $c\in[0,\infty)$, a polynomial $P$ and $T:=t+\delta$. Hence by Lebesgue's dominated convergence theorem we can show~\eqref{eq:psi-derivative-1-reformulation} by proving
\[ \lim_{h\to 0}\sup_{s\in[0,1]}\bigl\lv F_h'(t+sh)+\rmi V(x)F_0(t)\bigr\rv=0 \]
pointwisely, that is for all paths of Brownian motion separately. Fix a path of the Brownian motion. Due to~\eqref{eq:psi-derivative-1-derivative} and the fact that $\rmi V(x)F_0(t)$ is independent of $s$ and $h$ it suffices to show the two equalities
\begin{equation}\label{eq:psi-derivative-1-step1}
\lim_{h\to 0}\sup_{s\in[0,1]}\bigl\lv F_h(t+sh)-F_0(t)\bigr\rv=0
\end{equation}
and
\begin{equation}\label{eq:psi-derivative-1-step2}
\lim_{h\to 0}\sup_{s\in[0,1]}\Bigl\lv V\Bigl(y+\frac{t+sh}{t+h}(x-y)+\sqrt{\rmi}B_{t+sh}-\frac{t+sh}{t+h}\sqrt{\rmi}B_{t+h}\Bigr)-V(x)\Bigr\rv=0.
\end{equation}
For~\eqref{eq:psi-derivative-1-step1} note that by continuity of $\exp$ it suffices to show
\begin{multline*}
\lim_{h\to 0}\sup_{s\in[0,1]}\biggl\lv\int_0^{t+sh}V\Bigl(y+\frac{r}{t+h}(x-y)+\sqrt{\rmi}B_r-\frac{r}{t+h}\sqrt{\rmi}B_{t+h}\Bigr)\,\rmd r\\
-\int_0^tV\Bigl(y+\frac{r}{t}(x-y)+\sqrt{\rmi}B_r-\frac{r}{t}\sqrt{\rmi}B_t\Bigr)\,\rmd r\biggr\rv=0,
\end{multline*}
for which in turn we show the two equalities
\begin{equation}\label{eq:psi-derivative-1-step1.1}
\lim_{h\to 0}\sup_{s\in[0,1]}\biggl\lv\int_t^{t+sh}V\Bigl(y+\frac{r}{t+h}(x-y)+\sqrt{\rmi}B_r-\frac{r}{t+h}\sqrt{\rmi}B_{t+h}\Bigr)\,\rmd r\biggr\rv=0
\end{equation}
and
\begin{multline}\label{eq:psi-derivative-1-step1.2}
\lim_{h\to 0}\biggl\lv\int_0^tV\Bigl(y+\frac{r}{t+h}(x-y)+\sqrt{\rmi}B_r-\frac{r}{t+h}\sqrt{\rmi}B_{t+h}\Bigr)\,\rmd r\\
-\int_0^tV\Bigl(y+\frac{r}{t}(x-y)+\sqrt{\rmi}B_r-\frac{r}{t}\sqrt{\rmi}B_t\Bigr)\,\rmd r\biggr\rv=0.
\end{multline}
Clearly~\eqref{eq:psi-derivative-1-step1.1} follows from
\[ \sup_{\lv h\rv\le\delta}\sup_{r\in[t-\delta,t+\delta]}\Bigl\lv V\Bigl(y+\frac{r}{t+h}(x-y)+\sqrt{\rmi}B_r-\frac{r}{t+h}\sqrt{\rmi}B_{t+h}\Bigr)\Bigr\rv<\infty, \]
while~\eqref{eq:psi-derivative-1-step1.2} is shown by
\begin{align*}
0
& \le \lim_{h\to 0}\sup_{r\in[0,t]}\biggl\lv\frac{r}{t+h}(x-y)-\frac{r}{t+h}\sqrt{\rmi}B_{t+h}-\frac{r}{t}(x-y)+\frac{r}{t}\sqrt{\rmi}B_t\biggr\rv\\
& \le \lim_{h\to 0}\lv x-y\rv\sup_{r\in[0,t]}\biggl\lv\frac{r}{t+h}-\frac{r}{t}\biggr\rv + \sup_{r\in[0,t]}\biggl\lv\frac{r}{t+h}B_{t+h}-\frac{r}{t}B_t\biggr\rv = 0
\end{align*}
and the fact that $V$ is uniformly continuous on the compact set
\[ \biggl\{y+\frac{r}{t+h}(x-y)+\sqrt{\rmi}B_r-\frac{r}{t+h}\sqrt{\rmi}B_{t+h}:r\in[0,t],\,\lv h\rv\le\delta\biggr\}\subset\C. \]
Similarly~\eqref{eq:psi-derivative-1-step2} follows from continuity of $V$ and
\begin{align*}
0
& \le \sup_{s\in[0,1]}\biggl\lv\biggl(y+\frac{t+sh}{t+h}(x-y)+\sqrt{\rmi}B_{t+sh}-\frac{t+sh}{t+h}\sqrt{\rmi}B_{t+h}\biggr)-x\biggr\rv\\
& \le \lv x-y\rv\sup_{s\in[0,1]}\biggl\lv 1-\frac{t+sh}{t+h}\biggr\rv+\sup_{s\in[0,1]}\biggl\lv B_{t+sh}-\frac{t+sh}{t+h}B_{t+h}\biggr\rv\xrightarrow{h\to 0}0.
\end{align*}
Hence~\eqref{eq:psi-derivative-1-reformulation} is shown.
\end{proof}

\begin{lemma}\label{lem:psi-derivative-2}
For $t>0$ and $x,y\in\C$ it holds
\begin{multline*}
\lim_{h\to 0}\frac{1}{h}\E\biggl[\exp\biggl(-\rmi\int_0^tV\Bigl(y+\frac{r}{t+h}(x-y)+\sqrt{\rmi}B_r-\frac{r}{t+h}\sqrt{\rmi}B_{t+h}\Bigr)\,\rmd r\biggr)\\
-\exp\biggl(-\rmi\int_0^tV\Bigl(y+\frac{r}{t}(x-y)+\sqrt{\rmi}B_r-\frac{r}{t}\sqrt{\rmi}B_{t+h}\Bigr)\,\rmd r\biggr)\biggr]\\
=-\frac{x-y}{t}\frac{\rmd}{\rmd x}\E\biggl[\exp\biggl(-\rmi\int_0^tV\Bigl(y+\frac{r}{t}(x-y)+\sqrt{\rmi}B_r-\frac{r}{t}\sqrt{\rmi}B_t\Bigr)\,\rmd r\biggr)\biggr].
\end{multline*}
\end{lemma}

\begin{proof}
Again choose $\delta\in(0,t)$ and set $T:=t+\delta$. For $\lv h\rv\le\delta$, $u\in[t-\delta,t+\delta]$ and $x\in\C$ define the complex-valued random variable
\[ F_h(u,x):=\exp\biggl(-\rmi\int_0^tV\Bigl(y+\frac{r}{u}(x-y)+\sqrt{\rmi}B_r-\frac{r}{u}\sqrt{\rmi}B_{t+h}\Bigr)\,\rmd r\biggr). \]
By the Leibniz integral rule, the partial derivative with respect to $x$ can be computed as
\begin{equation}\label{eq:psi-derivative-2-derivative}
\partial_xF_h(u,x)=F_h(u,x)\int_0^t-\rmi\frac{r}{u}V'\Bigl(y+\frac{r}{u}(x-y)+\sqrt{\rmi}B_r-\frac{r}{u}\sqrt{\rmi}B_{t+h}\Bigr)\,\rmd r.
\end{equation}
Note that for compact $K\subset\C$ by Lemma~\ref{lem:uniform-polynomial-bound} there exists a polynomial $P$ such that
\[ \sup_{\lv h\rv\le\delta}\sup_{x\in K}\sup_{u\in[t-\delta,t+\delta]}\sup_{r\in[0,t]}\biggl\lv\frac{r}{u}V'\Bigl(y+\frac{r}{u}(x-y)+\sqrt{\rmi}B_r-\frac{r}{u}\sqrt{\rmi}B_{t+h}\Bigr)\biggr\rv\le P(\lV B\rV_T), \]
and thus, by Lemma~\ref{lem:boundedness-of-exponential-term}, there exists $c\in[0,\infty)$ such that
\begin{equation}\label{eq:psi-derivative-2-derivative-bound}
\sup_{\lv h\rv\le\delta}\sup_{x\in K}\sup_{u\in[t-\delta,t+\delta]}\lv\partial_xF_h(u,x)\rv\le\rme^{ct}\cdot t\cdot P(\lV B\rV_T)\in L^1(\mu).
\end{equation}
In particular by continuity of $\partial_xF_0(t,x)$ with respect to $x\in\R$ it holds
\[ -\frac{x-y}{t}\frac{\rmd}{\rmd x}\E\bigl[F_0(t,x)\bigr] = \E\Bigl[-\frac{x-y}{t}\partial_xF_0(t,x)\Bigr]
 = \E\Bigl[-\frac{x-y-\sqrt{\rmi}B_t}{t}\partial_xF_0(t,x)\Bigr], \]
where the second equality is due to independence of $B_t$ and $\bigl(B_r-\frac{r}{t}B_t\bigr)_{r\in[0,t]}$ and the fact that $\E[B_t]=0$. As in the proof of Lemma~\ref{lem:psi-derivative-1} the assertion follows if we can show
\begin{equation}\label{eq:psi-derivative-2-reformulation}
\lim_{h\to 0}\frac{1}{h}\bigl(F_h(t+h,x)-F_h(t,x)\bigr) = -\frac{x-y-\sqrt{\rmi}B_t}{t}\partial_xF_0(t,x)
\end{equation}
in $L^1(\mu)$. One easily computes
\begin{multline*}
\frac{\rmd}{\rmd u}V\Bigl(y+\frac{r}{u}(x-y)+\sqrt{\rmi}B_r-\frac{r}{u}\sqrt{\rmi}B_{t+h}\Bigr)\\
=-\frac{x-y-\sqrt{\rmi}B_{t+h}}{u}\frac{\rmd}{\rmd x}V\Bigl(y+\frac{r}{u}(x-y)+\sqrt{\rmi}B_r-\frac{r}{u}\sqrt{\rmi}B_{t+h}\Bigr),
\end{multline*}
for all $r\in[0,t]$, so with the Leibniz integral rule it follows that
\[ \partial_uF_h(u,x)=-\frac{x-y-\sqrt{\rmi}B_{t+h}}{u}\partial_xF_h(u,x) \]
is continuous in $u\in[t-\delta,t+\delta]$. The fundamental theorem of calculus yields
\begin{align*}
\frac{1}{h}\bigl(F_h(t+h,x)-F_h(t,x)\bigr)
& = \int_0^1\partial_uF_h(t+sh,x)\,\rmd s\\
& = \int_0^1-\frac{x-y-\sqrt{\rmi}B_{t+h}}{t+sh}\partial_xF_h(t+sh,x)\,\rmd s.
\end{align*}
Due to the fact that
\[ \sup_{\lv h\rv\le\delta}\sup_{s\in[0,1]}\biggl\lv\frac{x-y-\sqrt{\rmi}B_{t+h}}{t+sh}\partial_xF_h(t+sh,x)\biggr\rv\le \frac{\lv x-y\rv+\lV B\rV_T}{t-\delta}\cdot\rme^{ct}\cdot t\cdot P(\lV B\rV_T)\in L^1(\mu), \]
where $c$ and $P$ are as in~\eqref{eq:psi-derivative-2-derivative-bound} for the singleton $K=\{x\}$, by Lebesgue's dominated convergence theorem we can show~\eqref{eq:psi-derivative-2-reformulation} by proving that
\[ \lim_{h\to 0}\sup_{s\in[0,1]}\bigl\lv\partial_xF_h(t+sh,x)-\partial_xF_0(t,x)\bigr\rv=0 \]
holds pointwisely, that is for each fixed path of Brownian motion. Similarly as in the proof of Lemma~\ref{lem:psi-derivative-1}, due to~\eqref{eq:psi-derivative-2-derivative} and the fact that $\partial_xF_0(t,x)$ is independent of $s$ and $h$, we can prove this by separately showing
\begin{equation}\label{eq:psi-derivative-2-step1}
\lim_{h\to 0}\sup_{s\in[0,1]}\bigl\lv F_h(t+sh,x)-F_0(t,x)\bigr\rv=0
\end{equation}
and
\begin{multline}\label{eq:psi-derivative-2-step2}
\lim_{h\to 0}\sup_{s\in[0,1]}\biggl\lv \int_0^t\frac{r}{t+sh}V'\Bigl(y+\frac{r}{t+sh}(x-y)+\sqrt{\rmi}B_r-\frac{r}{t+sh}\sqrt{\rmi}B_{t+h}\Bigr)\,\rmd r\\
 -\int_0^t\frac{r}{t}V'\Bigl(y+\frac{r}{t}(x-y)+\sqrt{\rmi}B_r-\frac{r}{t}\sqrt{\rmi}B_t\Bigr)\,\rmd r\biggr\rv=0.
\end{multline}
The techniques to prove this have already been used in the proof of Lemma~\ref{lem:psi-derivative-1}. For~\eqref{eq:psi-derivative-2-step1} by continuity of $\exp$ it suffices to show
\begin{multline*}
\lim_{h\to 0}\sup_{s\in[0,1]}\biggl\lv\int_0^tV\Bigl(y+\frac{r}{t+sh}(x-y)+\sqrt{\rmi}B_r-\frac{r}{t+sh}\sqrt{\rmi}B_{t+h}\Bigr)\,\rmd r\\
-\int_0^tV\Bigl(y+\frac{r}{t}(x-y)+\sqrt{\rmi}B_r-\frac{r}{t}\sqrt{\rmi}B_t\Bigr)\,\rmd r\biggr\rv=0,
\end{multline*}
which follows from
\begin{align*}
0
& \le \sup_{s\in[0,1]}\sup_{r\in[0,t]}\biggl\lv\frac{r}{t+sh}(x-y)-\frac{r}{t+sh}\sqrt{\rmi}B_{t+h}-\frac{r}{t}(x-y)+\frac{r}{t}\sqrt{\rmi}B_t\biggr\rv\\
& \le \lv x-y\rv\sup_{s\in[0,1]}\sup_{r\in[0,t]}\biggl\lv\frac{r}{t+sh}-\frac{r}{t}\biggr\rv+\sup_{s\in[0,1]}\sup_{r\in[0,t]}\biggl\lv\frac{r}{t+sh}B_{t+h}-\frac{r}{t}B_t\biggr\rv\xrightarrow{h\to 0}0
\end{align*}
and the fact that $V$ is uniformly continuous on the compact set
\[ \biggl\{y+\frac{r}{t+sh}(x-y)+\sqrt{\rmi}B_r-\frac{r}{t+sh}\sqrt{\rmi}B_{t+h}:r\in[0,t],\,\lv h\rv\le\delta,\,s\in[0,1]\biggr\}\subset\C. \]
Equality in~\eqref{eq:psi-derivative-2-step2} is shown analogously.
\end{proof}

\begin{lemma}
For $t>0$ and $x,y\in\C$ it holds
\begin{multline*}
\lim_{h\downarrow 0}\frac{1}{h}\E\biggl[\exp\biggl(-\rmi\int_0^tV\Bigl(y+\frac{r}{t}(x-y)+\sqrt{\rmi}B_r-\frac{r}{t}\sqrt{\rmi}B_{t+h}\Bigr)\,\rmd r\biggr)\\
-\exp\biggl(-\rmi\int_0^tV\Bigl(y+\frac{r}{t}(x-y)+\sqrt{\rmi}B_r-\frac{r}{t}\sqrt{\rmi}B_t\Bigr)\,\rmd r\biggr)\biggr]\\
= \frac{\rmi}{2}\frac{\rmd^2}{\rmd x^2}\E\biggl[\exp\biggl(-\rmi\int_0^tV\Bigl(y+\frac{r}{t}(x-y)+\sqrt{\rmi}B_r-\frac{r}{t}\sqrt{\rmi}B_t\Bigr)\,\rmd r\biggr)\biggr].
\end{multline*}
Note that in contrast to Lemmas~\ref{lem:psi-derivative-1} and~\ref{lem:psi-derivative-2}, here we only compute a one-sided limit.
\end{lemma}

\begin{proof}
For notational reasons we write the expectations as integrals in this proof. For $\omega\in S'(\R)$ and $z\in\C$ define
\[ f(z,\omega):=\exp\biggl(-\rmi\int_0^tV\Bigl(y-\frac{r}{t}y+\sqrt{\rmi}B_r(\omega)-\frac{r}{t}\sqrt{\rmi}B_t(\omega)+\frac{r}{t}z\Bigr)\,\rmd r\biggr)\in\C. \]
We have to show
\[ \frac{\rmd^+}{\rmd^+ u}\int_{S'(\R)}f\Bigl(x-\sqrt{\rmi}(B_{t+u}(\omega)-B_t(\omega)),\omega\Bigr)\,\rmd\mu(\omega)\biggr\rv_{u=0}=\frac{\rmi}{2}\frac{\rmd^2}{\rmd x^2}\int_{S'(\R)} f(x,\omega)\,\rmd\mu(\omega). \]
For $u\ge 0$, since Brownian motion fulfills $\mathcal{L}(-(B_{t+u}-B_t))=\mathcal{L}(B_u)$ and has independent increments, it holds
\begin{multline*}
\int_{S'(\R)} f\Bigl(x-\sqrt{\rmi}(B_{t+u}(\omega)-B_t(\omega)),\omega\Bigr)\,\rmd\mu(\omega)\\
= \int_{S'(\R)}\int_{S'(\R)} f\Bigl(x+\sqrt{\rmi}\bigl(B_{t+u}(\omega_2)-B_t(\omega_2)\bigr),\omega_1\Bigr)\,\rmd\mu(\omega_2)\,\rmd\mu(\omega_1)\\
= \int_{S'(\R)}\int_{S'(\R)} f\Bigl(x+\sqrt{\rmi}B_u(\omega_2),\omega_1\Bigr)\,\rmd\mu(\omega_2)\,\rmd\mu(\omega_1).
\end{multline*}
The general strategy of the proof is to show the equalities
\begin{align}
& \frac{\rmd}{\rmd u}\int_{S'(\R)}\int_{S'(\R)} f\bigl(x+\sqrt{\rmi}B_u(\omega_2),\omega_1\bigr)\,\rmd\mu(\omega_2)\,\rmd\mu(\omega_1)\notag\\
& = \int_{S'(\R)}\frac{\rmd}{\rmd u}\int_{S'(\R)} f\bigl(x+\sqrt{\rmi}B_u(\omega_2),\omega_1\bigr)\,\rmd\mu(\omega_2)\,\rmd\mu(\omega_1)\label{eq:psi-derivative-3-step1}\\
& = \int_{S'(\R)}\frac{\rmi}{2}\frac{\rmd^2}{\rmd x^2}\int_{S'(\R)} f\bigl(x+\sqrt{\rmi}B_u(\omega_2),\omega_1\bigr)\,\rmd\mu(\omega_2)\,\rmd\mu(\omega_1)\label{eq:psi-derivative-3-step2}\\
& = \frac{\rmi}{2}\frac{\rmd^2}{\rmd x^2}\int_{S'(\R)}\int_{S'(\R)} f\bigl(x+\sqrt{\rmi}B_u(\omega_2),\omega_1\bigr)\,\rmd\mu(\omega_2)\,\rmd\mu(\omega_1)\notag
\end{align}
at $u=0$. First we show equality between~\eqref{eq:psi-derivative-3-step1} and~\eqref{eq:psi-derivative-3-step2}, for which we apply Theorem~\ref{thm:freedoss} to each of the functions $f(\cdot,\omega_1)$ for fixed $\omega_1\in S'(\R)$. Fix an arbitrary $T>t$. Clearly $f(\cdot,\omega_1)$ is analytic, since $V$ is a polynomial. Similar as in~\eqref{eq:psi-derivative-2-derivative-bound} for compact $K\subset\C$ there exists $c\in[0,\infty)$ and a polynomial $P$ such that for all $\omega_1,\omega_2\in S'(\R)$ we have
\begin{align*}
\sup_{x\in K}\sup_{u\in[0,T]}\bigl\lv f\bigl(x+\sqrt\rmi B_u(\omega_2),\omega_1\bigr)\bigr\rv & \le\rme^{ct}\\
\sup_{x\in K}\sup_{u\in[0,T]}\Bigl\lv\frac{\rmd}{\rmd x}f\bigl(x+\sqrt\rmi B_u(\omega_2),\omega_1\bigr)\Bigr\rv & \le\rme^{ct}\cdot t\cdot P\bigl(\lV B(\omega_1)\rV_T+\lV B(\omega_2)\rV_T\bigr),\\
\sup_{x\in K}\sup_{u\in[0,T]}\Bigl\lv\frac{\rmd^2}{\rmd x^2}f\bigl(x+\sqrt\rmi B_u(\omega_2),\omega_1\bigr)\Bigr\rv & \le\rme^{ct}\cdot t\cdot P\bigl(\lV B(\omega_1)\rV_T+\lV B(\omega_2)\rV_T\bigr),
\end{align*}
where $\lV B(\omega_i)\rV_T:=\sup_{u\in[0,T]}\lv B_u(\omega_i)\rv$ for $i=1,2$. In particular
\[ \sup_{u\in[0,T]}\bigl\lv f\bigl(x+\sqrt\rmi B_u(\omega_2),\omega_1\bigr)\bigr\rv\quad\text{and}\quad\sup_{u\in[0,T]}\Bigl\lv\frac{\rmd^2}{\rmd x^2} f\bigl(x+\sqrt\rmi B_u(\omega_2),\omega_1\bigr)\Bigr\rv \]
are integrable with respect to $\omega_2$ for all $x\in\C$. Finally the function
\[ \C\ni x\longmapsto\int_{S'(\R)} f\bigl(x+\sqrt{\rmi}B_u(\omega_2),\omega_1\bigr)\,\rmd\mu(\omega_2)\in\C \]
is analytic for all $u\ge 0$, since for closed curves $\gamma$ it holds
\[ \int_\gamma\int_{S'(\R)} f\bigl(x+\sqrt{\rmi}B_u(\omega_2),\omega_1\bigr)\,\rmd\mu(\omega_2)\,\rmd x=0 \]
by the theorems of Fubini and Morera. Here Fubini's theorem is applicable due to boundedness of the integrand. We have verified the assumptions of Theorem~\ref{thm:freedoss} on $f(\cdot,\omega_1)$, so equality between~\eqref{eq:psi-derivative-3-step1} and~\eqref{eq:psi-derivative-3-step2} is shown. Moreover, boundedness of the integrand also justifies each interchange of differentiation and integration.
\end{proof}

We have proven Proposition~\ref{prop:psi-derivative}. To see that the left-sided derivative exists as well and coincides with the right-sided one, we use the basic result that a continuous function $f\colon(0,T)\to\C$ which has a continuous one-sided derivative is continuously differentiable.

To use this fact, observe that for all $x,y\in\C$ the functions $\psi(t,x;y)$, $\partial_x\psi(t,x;y)$ and $\Delta\psi(t,x;y)$ are continuous in $t\in(0,T)$. As before, this is proven by continuity of the integrand and Lebesgue's dominated convergence theorem.

\begin{corollary}\label{cor:psi-derivative}
For all $x,y\in\C$ we have that
\[ \partial_t\psi(t,x;y)=\Bigl(-\rmi V(x)-\frac{x-y}{t}\partial_x+\frac{\rmi}{2}\Delta\Bigr)\psi(t,x;y) \]
exists and is continuous in $t>0$.
\end{corollary}

We close this section by proving that $K_V(t,x;y)$ is continuously differentiable with respect to $(t,x,y)\in(0,\infty)\times\R^2$, and we do so by showing that each partial derivative exists and is continuous itself. Since $K_V(t,x;y)=K_0(t,x;y)\cdot\psi(t,x;y)$ and $K_0(t,x;y)$ is infinitely differentiable, we only need to do this for $\psi(t,x;y)$. Due to $\psi(t,x;y)=\psi(t,y;x)$ it is sufficient to consider the partial derivatives with respect to $t$ and $x$.

In what follows we use the random variable
\begin{equation}\label{eq:def-Z}
Z_t(x,y):=-\rmi\int_0^tV\Bigl(y+\frac{r}{t}(x-y)+\sqrt{\rmi}B_r-\frac{r}{t}\sqrt{\rmi}B_t\Bigr)\,\rmd r
\end{equation}
for $t>0$ and $x,y\in\C$ so that $\psi(t,x;y)=\E[\exp(Z_t(x,y))]$.

\begin{lemma}\label{lem:Z-derivatives}
For all paths of Brownian motion and all $n\in\N$ we have that $\partial_x^nZ_t(x,y)$ exists for $(t,x,y)\in(0,\infty)\times\C^2$ and $\partial_x^nZ$ extends continuously to $[0,\infty)\times\C^2$ via $\partial_x^nZ_0(x,y):=0$. Furthermore, for compact $K\subset\C^2$ there exists a polynomial $P=P_K$ such that
\[ \sup_{n\in\N}\sup_{t\in[0,T]}\sup_{(x,y)\in K}\lv\partial_x^nZ_t(x,y)\rv\le T\cdot P(\lV B\rV_T) \]
for all $T>0$.
\end{lemma}

\begin{proof}
By the Leibniz integral rule we have
\[ \partial_x^nZ_t(x,y)=-\rmi\int_0^t\frac{r^n}{t^n}V^{(n)}\Bigl(y+\frac{r}{t}(x-y)+\sqrt{\rmi}B_r-\frac{r}{t}\sqrt{\rmi}B_t\Bigr)\,\rmd r \]
for all $t>0$ and $x,y\in\C$. Since $V$ is a polynomial it is obvious that this is a continuous function in $(t,x,y)\in(0,\infty)\times\C^2$. Lemma~\ref{lem:uniform-polynomial-bound} shows that for compact $K\subset\C^2$ there is a polynomial $P_n$ such that
\[ \sup_{(x,y)\in K}\lv\partial_x^nZ_t(x,y)\rv\le\int_0^tP_n(\lV B\rV_T)\,\rmd r=t\cdot P_n(\lV B\rV_T) \]
for $0<t\le T$. Thus defining $\partial_x^nZ_0(x,y):=0$ extends $\partial_x^nZ_t(x,y)$ continuously to $(t,x,y)\in[0,\infty)\times\C^2$. Since $V^{(n)}=0$ for $n>\deg V$ one single polynomial $P$ can be found which fulfills the desired estimate.
\end{proof}

\begin{proposition}\label{prop:psi-derivative-x}
For all $n\in\N$ we have that $\partial_x^n\psi(t,x;y)$ exists for $(t,x,y)\in(0,\infty)\times\C^2$ and $\partial_x^n\psi$ extends continuously to $[0,\infty)\times\C^2$.
\end{proposition}

\begin{proof}
Since $\partial_x^nZ_t(x,y)=0$ for $n>\deg V$, by the chain rule and the product rule for differentation it is obvious that for each $n\in\N$ there exists a polynomial $P_n\colon\C^{\deg V}\to\C$ such that
\[ \frac{\rmd^n}{\rmd x^n}\exp(Z_t(x,y))=P_n\bigl(\partial_xZ_t(x,y),\dotsc,\partial_x^{\deg V}Z_t(x,y)\bigr)\cdot\exp(Z_t(x,y)). \]
By Lemma~\ref{lem:Z-derivatives} this extends to a continuous function on $[0,\infty)\times\C^2$. Moreover, it is dominated by an integrable random variable locally uniformly by Lemmas~\ref{lem:boundedness-of-exponential-term} and~\ref{lem:doss-expectation-estimate}. Hence Lebesgue's dominated convergence theorem yields that
\[ \partial_x^n\psi(t,x;y)=\frac{\rmd^n}{\rmd x^n}\E\bigl[\exp\bigl(Z_t(x,y)\bigr)\bigr]=\E\Bigl[\frac{\rmd^n}{\rmd x^n}\exp\bigl(Z_t(x,y)\bigr)\Bigr] \]
extends continuously to $[0,\infty)\times\C^2$ as well.
\end{proof}

\begin{proposition}
For all $n\in\N$ it holds that $\partial_x^n\partial_t\psi$ exists and is continuous on $(0,\infty)\times\C^2$.
\end{proposition}

\begin{proof}
We already know
\[ \partial_t\psi(t,x;y)=\Bigl(-\rmi V(x)-\frac{x-y}{t}\partial_x+\frac{\rmi}{2}\Delta\Bigr)\psi(t,x;y), \]
hence by Proposition~\ref{prop:psi-derivative-x} it is clear that $\partial_t\psi$ is infinitely differentiable with respect to $x$ and the derivatives are continuous on $(0,\infty)\times\C^2$.
\end{proof}

\begin{corollary}\label{cor:KV-is-differentiable}
$K_V$ is continuously differentiable on $(0,\infty)\times\C^2$.
\end{corollary}

\end{subsection}

\begin{subsection}{Verifying the Schrödinger Equation}

In this section we show in which sense $K_V$ is a fundamental solution to the Schrödinger equation. Moreover we show that solutions to~\eqref{eq:schroedinger-equation} for initial states $u_0$ with compact support can be obtained via~\eqref{eq:def-solution-via-integral}.

The free particle propagator $K_0$ as in~\eqref{eq:free-propagator} is well-known to be a solution to the free Schrödinger equation with initial condition $\delta_x$. In fact
\begin{equation}\label{eq:K0-diracsequence}
\lim_{t\downarrow 0}\int_\R K_0(t,x;y)u_0(y)\,\rmd y=u_0(x)
\end{equation}
for all $x\in\R$ and $u_0\colon\R\to\C$ which are compactly supported and twice differentiable. We want to transfer this fact to $K_V$ and show that
\begin{equation}\label{eq:KV-diracsequence}
\lim_{t\downarrow 0}\int_\R K_V(t,x;y)u_0(y)\,\rmd y = u_0(x)
\end{equation}
is valid as well.

\begin{remark}
If $u_0$ does not have compact support, we cannot ensure existence of the integral $\int_\R K_V(t,x;y)u_0(y)\,\rmd y$ even when $t>0$ is fixed. For example, consider the rapidly decreasing function $u_0(y)=\rme^{-y^2}$ and the potential $V(z)=z^6$. Then integrability of
\[ K_V(t,x;y)u_0(y)=K_0(t,x;y)\cdot\E\biggl[\exp\biggl(-\rmi\int_0^tV\Bigl(y+\frac{r}{t}(x-y)+\sqrt{\rmi}B_r-\frac{r}{t}\sqrt{\rmi}B_t\Bigr)\,\rmd r-y^2\biggr)\biggr] \]
with respect to $y\in\R$ is rather questionable, also see~\cite[Rem.~7.26]{Vogel2010}. For this reason, we only consider initial states $u_0$ with compact support.
\end{remark}

To link~\eqref{eq:K0-diracsequence} and~\eqref{eq:KV-diracsequence} the following is useful.

\begin{lemma}\label{lem:uniform-propagator-convergence}
Let $K\subset\R$ be compact. Then
\[ \lim_{t\downarrow 0}\sup_{x,y\in K}\bigl\lv K_V(t,x;y)-K_0(t,x;y)\bigr\rv=0. \]
\end{lemma}

\begin{proof}
Recall that $K_V(t,x;y)=K_0(t,x;y)\cdot\psi(t,x;y)$. Since for $x,y\in\R$ it holds $\lv K_0(t,x;y)\rv=(2\pi t)^{-1/2}$ we have to show
\[ \lim_{t\downarrow 0}\frac{1}{\sqrt{t}}\sup_{x,y\in K}\lv\psi(t,x;y)-1\rv=0. \]
Let $Z_t(x,y)$ be as in~\eqref{eq:def-Z} so that $\psi(t,x;y)=\E[\exp(Z_t(x,y))]$. Then
\[ \sup_{x,y\in K}\lv\exp(Z_t(x,y))\rv\le\rme^{ct} \]
for some $c\in[0,\infty)$ and all $t>0$ by Lemma~\ref{lem:boundedness-of-exponential-term}. Moreover there exists a polynomial $P$ such that
\[ \sup_{x,y\in K}\lv Z_t(x,y)\rv\le t\cdot P(\lV B\rV_T) \]
whenever $0<t\le T$ as in the proof of Lemma~\ref{lem:Z-derivatives}. It follows
\begin{align*}
\sup_{x,y\in K}\lv\psi(t,x;y)-1\rv
 & \le \E\biggl[\sup_{x,y\in K}\lv\exp(Z_t(x,y))-1\rv\biggr]\\
 & \le \E\biggl[\sup_{x,y\in K}\lv Z_t(x,y)\rv\cdot\max\bigl\{1,\lv\exp(Z_t(x,y))\rv\bigr\}\biggr]\\
 & \le \rme^{ct}\cdot t\cdot\E[P(\lV B\rV_T)],
\end{align*}
where we used that
\[ \lv\exp(z)-1\rv=\biggl\lv z\int_0^1\exp(sz)\,\rmd s\biggr\rv\le\lv z\rv\sup_{s\in[0,1]}\lv\exp(sz)\rv=\lv z\rv\max\{1,\lv\exp(z)\rv\} \]
holds for all $z\in\C$.
\end{proof}

\begin{theorem}\label{thm:KV-is-fundamental-solution}
$K_V$ solves the Schrödinger equation and fulfills $\lim_{t\downarrow 0}K_V(t,x;y)=\delta_x(y)$ in the sense that~\eqref{eq:KV-diracsequence} holds for all $x\in\R$ and $u_0\colon\R\to\C$ which are compactly supported and twice differentiable.
\end{theorem}

\begin{proof}
It is clear that $\partial_tK_0(t,x;y)=\frac{\rmi}{2}\Delta K_0(t,x;y)$ and $\partial_xK_0(t,x;y)=\rmi\frac{x-y}{t}K_0(t,x;y)$. Moreover $\partial_t\psi(t,x;y)$ is given in Corollary~\ref{cor:psi-derivative}. Applying the product rule for differentiation to $K_V(t,x;y)=K_0(t,x;y)\cdot\psi(t,x;y)$ shows that $K_V(t,x;y)$ indeed fulfills $\rmi\partial_tK_V(t,x;y)=-\frac{1}{2}\Delta K_V(t,x;y)+V(x)K_V(t,x;y)$ for $t>0$ and $x,y\in\C$. If $u_0\colon\R\to\C$ is compactly supported and twice differentiable, it is known that~\eqref{eq:K0-diracsequence} holds. Due to Lemma~\ref{lem:uniform-propagator-convergence}, also~\eqref{eq:KV-diracsequence} holds.
\end{proof}

To take an initial state $u_0\colon\R\to\C$ into account we use the ansatz~\eqref{eq:def-solution-via-integral}.

\begin{theorem}\label{thm:u-solves-schroedinger-equation}
Let $u_0\colon\R\to\C$ be compactly supported and twice differentiable. Then $u$ as in~\eqref{eq:def-solution-via-integral} solves the Schrödinger equation with initial condition $\lim_{t\downarrow 0}u(t,x)=u_0(x)$.
\end{theorem}

\begin{proof}
Let $a<b$ such that $u_0$ vanishes outside of $[a,b]$. Then
\begin{multline*}
\rmi\partial_tu(t,x)=\int_a^b\rmi\partial_tK_V(t,x;y)u_0(y)\,\rmd y\\
=\int_a^b\Bigl(-\frac{1}{2}\Delta+V(x)\Bigr)K_V(t,x;y)u_0(y)\,\rmd y=\Bigl(-\frac{1}{2}\Delta+V(x)\Bigr)u(t,x)
\end{multline*}
for all $t>0$ and $x\in\R$. Here interchange of differentiation and integration is possible due to continuity of the partial derivatives of $K_V$ and Lebesgue's dominated convergence theorem. Hence $u$ solves~\eqref{eq:schroedinger-equation}. Moreover $\lim_{t\downarrow 0}u(t,x)=u_0(x)$ holds for all $x\in\R$ by Theorem~\ref{thm:KV-is-fundamental-solution}.
\end{proof}

\begin{remark}
If $u_0$ is a finite linear combination of indicator functions of bounded intervals, then~\eqref{eq:K0-diracsequence} and~\eqref{eq:KV-diracsequence} hold for those points $x\in\R$ at which $u_0$ is continuous. This is due to the fact that for $-\infty<a<b<\infty$ it holds
\[ \lim_{t\downarrow 0}\int_\R K_0(t,x;y)\indi_{[a,b]}(y)\,\rmd y=\begin{cases}1 & \text{if }x\in(a,b),\\\frac{1}{2} & \text{if } x\in\{a,b\},\\0 & \text{if }x\not\in[a,b].\end{cases} \]
Consequently, a solution to~\eqref{eq:schroedinger-equation} is given as in~\eqref{eq:def-solution-via-integral} with an initial condition that coincides with $u_0$ in all but its finitely many points of discontinuity.
\end{remark}

Aside from the pointwise convergence in~\eqref{eq:K0-diracsequence} it is also known that
\[ \lim_{t\downarrow 0}\int_\R K_0(t,\cdot;y)u_0(y)\,\rmd y=u_0 \]
in $L^2(\R,\rmd x;\C)$ whenever $u_0\in L^1(\R,\rmd x;\C)\cap L^2(\R,\rmd x;\C)$. In particular we have weak convergence, namely
\begin{equation}\label{eq:weak-convergence-of-free-propagator}
\lim_{t\downarrow 0}\biggl(\int_\R K_0(t,\cdot;y)u_0(y)\,\rmd y,v\biggr)_{L^2(\R,\rmd x;\C)}=(u_0,v)_{L^2(\R,\rmd x;\C)}
\end{equation}
for all $v\in L^2(\R,\rmd x;\C)$. A weaker version of the latter may also be proven for $K_V(t,x;y)$. However, we do not know whether $\int_\R K_0(t,\cdot;y)u_0(y)\,\rmd y\in L^2(\R,\rmd x;\C)$, so we may not write the integral as an inner product.

\begin{proposition}
For all $u_0,v\colon\R\to\C$ which are bounded, measurable and compactly supported we have that
\[ \lim_{t\downarrow 0}\int_\R\int_\R K_V(t,x;y)u_0(y)\,\rmd y\,v(x)\,\rmd x=\int_\R u_0(x)v(x)\,\rmd x. \]
\end{proposition}

\begin{proof}
Due to Lemma~\ref{lem:uniform-propagator-convergence} it holds
\[ \lim_{t\downarrow 0}\int_\R\int_\R\bigl(K_V(t,x;y)-K_0(t,x;y)\bigr)u_0(y)\,\rmd y\,v(x)\,\rmd x=0. \]
Then the statement follows from~\eqref{eq:weak-convergence-of-free-propagator}.
\end{proof}

\begin{remark}
The class of admissible initial conditions considered in~\cite{Doss1980} consisted purely of analytic functions, which does not cover common initial states with compact support. We have shown in Theorem~\ref{thm:u-solves-schroedinger-equation} that solutions to~\eqref{eq:schroedinger-equation} are obtained by~\eqref{eq:def-solution-via-integral} for a large class of relevant non-analytic initial conditions.
\end{remark}

\begin{remark}
In~\cite{Yajima1996} it was shown that if the potential $V$ is super-quadratic, the fundamental solution to the Schrödinger equation~\eqref{eq:schroedinger-equation} is nowhere continuously differentiable. Since our class contains the super-quadratic $V(z)=z^6$, this result and differentiability of $K_V$ seem to explicitly contradict each other. However, in~\cite{Yajima1996} the term fundamental solution is understood as the distribution kernel of the unitary group generated by the unique self-adjoint extension of the operator $H$ on $L^2(\R,\rmd x;\C)$ defined by $Hf:=-\frac{1}{2}\Delta f+V\cdot f$ for infinitely differentiable functions $f\colon\R\to\C$ with compact support. Here rather than the operator semigroup approach we used path integrals to construct $K_V$, which generates a different solution.
\end{remark}

\end{subsection}

\end{section}

\bigskip

\noindent \textbf{Acknowledgment:}
We thank Anna Vogel who provided many ideas for computing the time derivative of $K_V$ in her Ph.D.\ thesis, see~\cite[Sec.~5.3]{Vogel2010}. The second author thanks both the German state Rhineland-Palatinate and the department of Mathematics at the University of Kaiserslautern for financial support in the form of a fellowship.

\end{document}